\documentclass[11pt,leqno]{amsart}
\usepackage{amsmath}
\usepackage{amssymb}
\usepackage{a4wide}
\usepackage{bm}
\usepackage{graphics}
\usepackage{hyperref}
\usepackage{color}
\definecolor{ao(english)}{rgb}{0.0, 0.5, 0.0}
\hypersetup{%
  colorlinks = true,
  linkcolor  = blue,
  citecolor    = ao(english)
}
\usepackage{comment}
 \usepackage[applemac]{inputenc}
\usepackage{latexsym, amscd, amsthm,amsfonts,amstext}
\usepackage[mathscr]{eucal}
\usepackage{appendix}
\usepackage[active]{srcltx}%goi lenh de sua tu dvi sang tex
\include{srctex}%goi lenh de sua tu dvi sang tex

 \usepackage{mleftright}
\usepackage{pgf,tikz}
\usetikzlibrary{patterns,arrows,intersections, decorations.markings, decorations.pathmorphing,
	backgrounds, positioning, fit, shapes.geometric}
\usepackage{caption}  
\usepackage{mathrsfs}
\usepackage{xparse}
\makeatletter

\RenewDocumentCommand{\title}{om}{%
   \IfNoValueTF{#1}
     {\gdef\shorttitle{Resonance asymptotics with degenerate crossings}}% 
     {\gdef\shorttitle{#1}}%
   \gdef\@title{#2}%
}

\DeclareFontFamily{U}{mathb}{\hyphenchar\font45}
\DeclareFontShape{U}{mathb}{m}{n}{ <-6> matha5 <6-7> matha6 <7-8>
mathb7 <8-9> mathb8 <9-10> mathb9 <10-12> mathb10 <12-> mathb12 }{}
\DeclareSymbolFont{mathb}{U}{mathb}{m}{n}

\DeclareMathAccent{\abxring}{0}{mathb}{"38}

\DeclareFontFamily{U}{mathb}{\hyphenchar\font45}
\DeclareFontShape{U}{mathb}{m}{n}{ <-6> matha5 <6-7> matha6 <7-8>
mathb7 <8-9> mathb8 <9-10> mathb9 <10-12> mathb10 <12-> mathb12 }{}
\DeclareSymbolFont{mathb}{U}{mathb}{m}{n}

\usepackage{subcaption} 

\graphicspath{{./Figures/}}

\newcommand{\N}{\mathbb{N}}
\newcommand{\Z}{\mathbb{Z}}
\newcommand{\R}{\mathbb{R}}
\newcommand{\C}{\mathbb{C}}
\newcommand{\W}{{\mathcal W}}

\newcommand{\cA}{{\mathcal A}}
\newcommand{\cR}{\mathcal{R}}

\newcommand{\p}{\partial}

\newcommand{\dl}{\delta}
\newcommand{\pphi}{\varphi}

\newcommand{\til}{\widetilde}

\newcommand{\re}{{\rm Re}\hskip 1pt }
\newcommand{\im}{{\rm Im}\hskip 1pt }
\newcommand{\ord}{{\mathcal O}}

\newcommand{\ope}[1]{\operatorname{#1}}
\newcommand{\mc}[1]{\mathcal{#1}}

\newcommand{\qtext}[1]{\quad\text{#1 }\ }

\newcommand{\be}{\begin{equation}}
\newcommand{\ee}{\end{equation}}
\newcommand{\ben}{\begin{equation*}}
\newcommand{\een}{\end{equation*}}

\newcommand{\inc}{{%\ope{in}}}
\flat}}
\newcommand{\out}{{%\ope{out}}}
\sharp}}
\newcommand{\dir}{{\bullet}}
\newcommand{\edg}{{\bm{e}}}

\newcommand{\boundellipse}[3]% center, xdim, ydim
{(#1) ellipse (#2 and #3)
}

\makeatletter
 \@addtoreset{equation}{section}
 \makeatother
 
 \usepackage{hyperref} 

\theoremstyle{theorem} 
\newtheorem{theorem}{Theorem}%[section]
\newtheorem{lemma}{Lemma}[section]
\newtheorem*{example}{Example}
\newtheorem{proposition}[lemma]{Proposition}
\newtheorem{remark}[lemma]{Remark}
\newtheorem{corollary}[lemma]{Corollary}
\newtheorem{As}{Assumption}

\theoremstyle{definition} 
\newtheorem{definition}[lemma]{Definition}
\numberwithin{equation}{section}
  {\setlength{\baselineskip}{1.5\baselineskip}

\title{Semiclassical resonance asymptotics for systems with degenerate 
crossings of classical trajectories %I: Type A crossings
}
\author{Marouane Assal}
\address{Marouane Assal, Univ. Bordeaux, CNRS, Bordeaux INP, IMB, UMR 5251, F-33400 Talence, France. e-mail: marouane.assal@math.u-bordeaux.fr} 

\author{Setsuro Fujiie}
\address{Setsuro Fujiie, Department of Mathematical Sciences, 
Ritsumeikan University, 111 Noji-Higashi, Kusatsu, 525-8577,  Japan. 
e-mail: fujiie@fc.ritsumei.ac.jp}
\author{Kenta Higuchi}
\address{Kenta Higuchi, Graduate School of Science and Engineering, Ehime University/ Bunkyocho 3, Matsuyama, Ehime, 790-8577, Japan.
 e-mail: higuchi.kenta.vf@ehime-u.ac.jp}

\begin{document}

\maketitle

\begin{abstract}
This paper is concerned with the asymptotics of resonances in the semiclassical limit $h\to 0^+$ for $2\times 2$ matrix Schr\"odinger operators in one dimension. We study the case where the two underlying classical Hamiltonian trajectories cross tangentially in the phase space.
In the setting that one of the classical trajectories is a simple closed curve whereas the other one is non-trapping, we show that the imaginary part of the resonances is of order $h^{(m_0+3)/(m_0+1)}$, where $m_0$ is the maximal contact order of the crossings.
This principal order comes from the subprincipal terms of the transfer matrix at crossing points which describe the propagation of microlocal solutions from one trajectory to the other. In addition, we compute explicitly the leading coefficient of the resonance widths in terms of the probability amplitudes associated with all the \textit{generalized classical trajectories} escaping to infinity from the closed trajectory.

\end{abstract}
%\tableofcontents

\section{Introduction}
In this work, we are interested in the asymptotic distribution of quantum resonances in the semiclassical limit for matrix Schr\"odinger operators, especially in the effect of  crossings of the underlying classical trajectories. 
We study the one-dimensional model of Schr\"odinger systems arising in the context of the Born-Oppenheimer approximation of diatomic molecular Hamiltonians:
\begin{equation}\label{System0} P = P(h):= 
\begin{pmatrix}
 P_1 & h U\\
h U^* & P_2
\end{pmatrix}.
\end{equation} 
%on $L^2(\mathbb R; \mathbb C^2)$, 
The diagonal part consists of scalar Schr\"odinger operators 
\ben
P_j=P_j(h):=(h D_x)^2  + V_j(x) \quad (j=1,2), \;\;\;\; D_x:= -i \frac d{dx},
\een
with smooth real-valued potentials $V_1$ and $V_2$ on $\mathbb R$, and the anti-diagonal part describing the interaction between $P_1$ and $P_2$ consists of ($h$ times) a first order semiclassical differential operator $U$ and its adjoint $U^*$ with smooth real-valued coefficients on $\mathbb R$. 
We consider the situation where, at an energy-level $E_0$, the classical trajectories of $p_1$ are trapping while the ones of $p_2$ are non-trapping. 
Here $p_j(x,\xi):= \xi^2+ V_j(x)$ is the classical Hamiltonian associated with $P_j$ $(j=1,2)$, and its classical trajectories are the integral curves of the corresponding Hamiltonian vector field $H_{p_j}=\p_\xi p_j\p_x-\p_x p_j\p_\xi$. In our one-dimensional framework, the classical trajectories of $p_j$ coincide with the characteristic set $\Gamma_j(E_0):= p_j^{-1}(E_0)$.  
The interaction between $P_1$ and $P_2$ may shift the real eigenvalues near $E_0$ of $P_1$ to the lower half complex plane as resonances of $P$ (Fermi's golden rule). 
The semiclassical asymptotics of the imaginary parts (widths) of the resonances has been studied in this setting in \cite{Ma,Na,Ba,As} and \cite{FMW1,FMW2,FMW3}. 
Under a simple well condition on $P_1$, the first-mentioned series of works studied the exponential decay rate in crossing free cases while the second one gave polynomial asymptotics in the presence of crossings.

The aim of this paper is to generalize \cite{FMW3}. In that paper, a simplest model where the classical trajectories of $p_1$ and $p_2$ cross transversally in the phase space $T^*\mathbb R= \R_x\times \R_\xi$ is considered. The asymptotic of the resonance widths is explicitly computed at the principal level. In particular, its polynomial power in $h$ is 2, which is determined only by the transversality of the crossings. 
In the present paper, we relax the transversality condition on crossings, and give a precise asymptotics of the resonance widths for degenerate crossings. In this case, we prove that the polynomial power of the principal asymptotics of the resonance width is $(m+3)/(m+1)$ (Theorem~\ref{MAINTH}), where $m$ stands for the contact order \eqref{CCOO} of the crossings (including the transversal case $m=1$). The number of crossing points is allowed to be finitely many. In that case, $m$ should be the maximum contact order among them.

The power $(m+3)/(m+1)$ is determined by the microlocal transfer matrix at each crossing point. The transfer matrix is a $2\times2$ matrix describing the microlocal outgoing data in terms of the incoming data, and we show in Theorem~\ref{thm:ConnectionAllowed} that it is approximated by $I_2+h^{1/(m+1)}T_{\ope{sub}}$ with an anti-diagonal matrix $T_{\ope{sub}}$. In the transversal case,  this asymptotic formula was computed in \cite{FMW3} for matrix Schr\"odinger operators and generalized in \cite{AsFu} for general pseudodifferential systems by applying the normal form reduction due to Helffer-Sj\"ostrand \cite{HeSj3} (see also \cite{Sj1,Cdv3}). 
In our degenerate case, such a reduction is not valid and we develop a new approach which brings the study of microlocal solutions to the system near a given crossing point $\rho=(\rho_x,\rho_\xi)\in \Gamma_1\cap \Gamma_2$ to that of exact local solutions near $\rho_x$ constructed by the method established in \cite{FMW1}.
 Microlocally near each classical trajectory, incoming to or outgoing from a crossing point, we compute the asymptotic behavior of these solutions. The principal term of the anti-diagonal entries is expressed as the inner product of the WKB solutions to $P_1$ and $P_2$, and its asymptotic expansion is computed by the stationary phase method. The critical points of this integral are the crossing points, and the degeneracy of the critical point coincides with the contact order of the crossing point. From this, we obtain the exponent $1/(m+1)$ of the anti-diagonal entries. 
Here we only consider the generic case, where the energy level is different from the crossing level. 
Otherwise, a crossing point is also a turning point, and the WKB approximation is not valid there. In such a case the analysis would be different.

A Bohr-Sommerfeld type ``microlocal quantization condition" for resonances is given using the above microlocal transfer matrices. Resonances will be approximated by energies for which the ``monodromy matrix" has one as eigenvalue (we call them pseudo-resonances). 
This is an extension of the simple well problem to our system, where a non-trapping trajectory $\Gamma_2$ exists in addition to the closed curve $\Gamma_1$. The usual Bohr-Sommerfeld quantization condition for the scalar operator $P_1$ is the condition that the multiplication factor $-\exp(2i\cA(E)/h)$ after a continuation along $\Gamma_1$ of a microlocal solution equals one, where $2\cA(E)$ is the action of $\Gamma_1$. 
We justify this approximation by pseudo-resonances in the following way. The non-existence of resonance outside an $\ord(h^\infty)$-neighborhood of the pseudo-resonances is proved by applying the argument in \cite{BFRZbook}.  The one-to-one correspondence between the pseudo-resonances and the resonances is shown using the Kato-Rellich theorem. This microlocal method enables us to consider not only a simple model but also general cases with finitely many crossings.

To obtain the principal term of the widths of pseudo-resonances, one needs to compute the asymptotic expansion of the microlocal transfer matrix up to  $\ord(h^{2/(m+1)})$. Instead, we employ the following additional argument. 
The width of a resonance can be expressed as $h$ times the square of the amplitude of the corresponding suitably normalized resonant state near infinity.  The asymptotic behavior of the resonant state coincide with that of the microlocal WKB asymptotics on the outgoing trajectories tending to infinity. We will see in Lemma~\ref{lem:PAC} that, for each outgoing trajectory tending to infinity,  this microlocal behavior is given as the sum of ``probability amplitude" over all ``paths" from $\Gamma_1$ to the trajectory. Here, paths are defined by crossing points and classical trajectories as in the graph theory (see Section~\ref{sec:prfMain}). The probability amplitude associated with a path is defined in Definition~\ref{def:PA} in terms of the actions of the classical trajectories and entries of the microlocal transfer matrices at the crossing points. 

The paper is organized as follows. The next section is devoted to the precise assumptions and the statement of our main result. In section \ref{SEC3}, we establish a microlocal connection formula at a crossing point in the phase space (Theorems \ref{thm:T-exists} and \ref{thm:ConnectionAllowed}).
In the last section, we define the monodromy matrix and pseudo-resonances by applying the transfer matrix given in the previous section, and show the asymptotic distribution of resonances (Theorem~\ref{MAINTH}) in relation with that of pseudo-resonances.

\section{Assumptions and main result}\label{sec:results}  
\subsection{Precise assumptions} Let $p_j(x,\xi):= \xi^2+ V_j(x)$, $(x,\xi)\in T^*\mathbb R$, be the classical Hamiltonian associated with $P_j$, and let $H_{p_j}:=2\xi\p_x-V_j'(x)\p_\xi$ be the corresponding Hamiltonian vector-field, $j=1,2$. For an energy $E\in \mathbb R$, denote by $\Gamma_j(E)$ the characteristic set associated with $P_j$ given by 
\ben
\Gamma_j(E):=p_j^{-1}(E)=\{(x,\xi)\in T^*\R;\,p_j(x,\xi)=E\} \quad (j=1,2),
\een
and set 
\ben
\Gamma(E):= \Gamma_1(E) \cup \Gamma_2(E).
\een
Introduce the interaction operator 
\ben
U=U(x,hD_x) := r_0(x) + i  r_1(x) h D_x.
\een
Throughout the paper, we assume the following condition.

\begin{As}\label{H1}
$V_1,V_2,r_0$ and $r_1$ are smooth real-valued functions on $\R$. They can be extended to bounded holomorphic functions on an angular complex domain near infinity of the form
\ben
\Sigma=\Sigma(\theta_0,R_0):=\{x\in\C;\,\left|\im x\right|<(\tan\theta_0)\left|\re x\right|,\ \left|\re x\right|>R_0\}
\een
for some constants $0<\theta_0<\pi/2$ and $R_0\ge0$.  
Moreover, for each $j=1,2$, $V_j$ 
%is analytic near $V_j^{-1}(E_0)$, and 
admits limits $V_j^\pm$ as $\re x\to \pm\infty$ in $\Sigma$.
%, with $\lim_{\re x\to \pm\infty, x\in\Sigma} V_j(x)\neq E_0$.
%$E_0$ is different from any one of these limits.
\end{As}
%Each limit is different from $E_0$. 
%, and satisfy
%\ben
%\lim_{\re x\to\pm\infty,\,x\in\Sigma}V_1(x)>E_0,\quad 
%\lim_{\re x\to\pm\infty,\,x\in\Sigma}V_2(x)\neq E_0.
%\een

In the sequel we fix a reference energy-level $E_0\in \mathbb R$ such that $E_0 \neq V_j^\pm$, $j=1,2$. For $\delta_1,\delta_2>0$ possibly $h$-dependent, one introduces the complex box 
\be\label{DEFR}
\cR=\cR_{E_0}(\delta_1,\delta_2) :=[E_0-\delta_1,E_0+\delta_1] +i[-\delta_2,\delta_2].
\ee
Under Assumption~\ref{H1}, the operator $P$ is self-adjoint in $L^2(\mathbb R;\mathbb C^2)$ with domain $H^2(\mathbb R;\mathbb C^2)$. One can define the resonances of $P$ in $\cR$, for sufficiently small $\dl_1$ and $\dl_2$, e.g., 
as the values $E\in\C$ such that the equation $Pw=Ew$ has a non-trivial outgoing solution $w$, that is, a non-identically vanishing solution such that for some $\theta>0$ sufficiently small, the function $x\mapsto w(\zeta_\theta(x))$ is in $L^2(\R;\C^2)$ (see e.g. \cite{Hu}), where %. For $0<\theta<\theta_0$, let 
$\zeta_\theta\in C^\infty(\R;\C)$ is such that $\zeta_\theta(x) = x$ for $ \left|x\right|\le R_0$ and $\zeta_\theta(x) = xe^{i\theta}$ for $ |x|\ge2R_0$. 
Equivalently, the resonances can be defined as the eigenvalues of a non-self adjoint  operator $P_\theta$ obtained from $P$ by the method of analytic distortion (see Section~\ref{sec:prfMain}). 
\begin{comment}
by the complex distorsion method due to Hunziker \cite{Hu} as follows (see also \cite[Section 2.7]{DyZw}).  
For $0<\theta<\theta_0$, let $\zeta_\theta\in C^\infty(\R;\C)$ be such that $\zeta_\theta(x) = x$ for $ \left|x\right|\le R_0$ and $\zeta_\theta(x) = xe^{i\theta}$ for $ |x|\ge2R_0$. 
We introduce the distortion operator 
\ben
\mathcal{U}_\theta u:=\left|\zeta_\theta'(x)\right|^{1/2}u\circ\zeta_\theta, 
\een
for analytic functions $u$. 
Consider the distorted operator 
\ben
P_{\theta}:= \mathcal{U}_\theta P ( \mathcal{U}_\theta)^{-1}.
\een 
This is first defined for analytic functions in $\Sigma$, and then extended to an operator on $L^2(\R;\C^2)$ by using the analyticity of the coefficients (Assumption~\ref{H1}). 
The operator $P_\theta$ is non-self-adjoint in $L^2(\mathbb R;\mathbb C^2)$, and its spectrum is discrete in $\cR$ due to the assumption $E_0 \neq V_j^\pm$, $j=1,2$ (see \cite[Appendix B]{Hi1}), and is independent of $\theta$. The resonances of $P$ are by definition the eigenvalues of $P_\theta$ in $\cR$. 
\end{comment}
In particular, the imaginary part of each resonance is negative.  We denote $\ope{Res}(P)$ the set of resonances of $P$.

We now make some assumptions on the geometry of  the characteristic set $\Gamma(E_0)$ and on the set of crossing points  
\ben
\mathcal{V}=\mathcal{V}(E_0) :=\Gamma_1(E_0)\cap\Gamma_2(E_0).
\een

\begin{As}\label{H2}
\begin{itemize}
\item[(i)] There exist real numbers $a_0<b_0$ such that 
\ben
\frac{V_1(x)-E_0}{(x-a_0)(x-b_0)}>0\quad \forall x\in \R.
\een
\item[(ii)]   The classically allowed region $\{x\in \mathbb R;\, V_2(x)\le E_0\}$ consists of unbounded intervals, and $V_2'(x)\ne 0$ for all $x\in \{x\in \mathbb R; V_2(x)=E_0\}$.
%\item[(iii)]  \ml $V_2(a_0)\neq E_0$ and $V_2(b_0)\neq E_0$. 
\end{itemize}
\end{As}
%For each $j=1,2$, $V_j'$ does not vanish on $V_j^{-1}(E_0)$, and $V_j$ is also analytic near each \textcolor{blue}{turning point. }
%There exist \textcolor{blue}{separate $V_1$ and $V_2$?}

Assumption \ref{H2} (i) means that $V_1$ admits a simple well at $E_0$   and hence $\Gamma_1(E_0)$ is a simple closed curve symmetric with respect to the $x$-axis, while (ii) implies that $E_0$ is a non-trapping energy for $p_2$.

For $E$ near $E_0$, we define the action integral 
%let $a=a(E)$ and $b=b(E)$ be the zeros of $V_1(x)-E$ near $a_0$ and $b_0$ respectively. In particular  $a(E_0) = a_0$ and $b(E_0)=b_0$. Let us define the action integral 
\ben\label{AAC}
\cA(E):=  \int_{\Gamma_1(E)} \xi dx.
\een
This is a smooth function near $E_0$. Under Assumption \ref{H2} (i), it is well known that the spectrum of the scalar Schr\"odinger operator $P_1$ in a small neighborhood of $E_0$ consists of eigenvalues approximated by the roots of the Bohr-Sommerfeld quantization rule (see e.g., \cite{ILR, Ya}) 
\begin{equation}\label{BSR}
\cos\left( \frac{\mathcal{A}(E)}{2h}\right) = 0.
\end{equation}
For $\delta_1,h>0$ small enough, we define an $h$-dependent discrete set $\mathfrak{B}_h\subset [E_0-\delta_1, E_0+\delta_1]$ by 
\begin{equation}\label{DEFFF}
\mathfrak{B}_h=\mathfrak{B}_{h}(\delta_1) := \left\{E\in [E_0-\delta_1,E_0+\delta_1];\, E\  \text{satisfies}\, \eqref{BSR} \right\}.
\end{equation}
The distance of two successive points in $\mathfrak{B}_h$ is of order $\mathcal{O}(h)$, more precisely, if $\delta_1= \mathcal{O}(h)$, then for any $E\in \mathfrak{B}_h$ one has 
\ben
\ope{dist}(\mathfrak{B}_h\setminus\{E\},E)=\frac {2\pi h}{\cA'(E_0)}+\ord(h^2).
\een

\begin{As}\label{H3}
The set of crossing points $\mathcal{V}$ is non-empty, and $\mathcal{V}\cap \{\xi=0\}=\emptyset$. Moreover, the contact order of $\Gamma_1(E_0)$ and $\Gamma_2(E_0)$ at each crossing point $\rho\in \mathcal{V}$ defined by
\begin{equation}\label{CCOO}
m_\rho:= \min \left\{m\in \mathbb N; \, H_{p_1}^{m}p_2(\rho)\neq0\right\}=\min \left\{m\in \mathbb N; \, H_{p_2}^{m}p_1(\rho)\neq0\right\}
\end{equation}
is finite.
\end{As}

Notice that since $\Gamma_1(E_0)$ is compact, the number of crossing points is finite under the above assumption. The condition $\mathcal{V}\cap \{\xi=0\}=\emptyset$ means the non-existence of a crossing-turning point. 
It is equivalent to $V_2(a_0)\neq E_0$ and $V_2(b_0)\neq E_0$. Since $\rho=(\rho_x, \rho_\xi)\in \mathcal{V}$ implies $\breve\rho:=(\rho_x, -\rho_\xi)\in \mathcal{V}$, the crossing points appear in pairs. 
The contact order $m_\rho$ at a crossing point $\rho=(\rho_x, \rho_\xi)\in \mathcal{V}$ coincides with the contact order of the potentials $V_1$ and $V_2$ at $\rho_x$, more precisely one has 
\begin{equation}\label{MRO}
m_\rho=
\min\{k\in\N; V_1^{(k)}(\rho_x)\ne V_2^{(k)}(\rho_x)\} .
\end{equation}
In fact, for any $m\in \mathbb N$ and $(x,\xi)\in \mathbb R^2$, one has 
 $$
H_{p_1}^m p_2 (x,\xi) = 2^m (V_2^{(m)}(x)-V_1^{(m)}(x)) \xi^m + \sum_{k=0}^{m-1} \alpha_{k,m}(x) \xi^k,
$$
with $\alpha_{k,m}$ of the form $\alpha_{k,m}(x) = \sum_{j=0}^{m-1} (V_2^{(j)}(x) - V_1^{(j)}(x)) \beta_{j,k,m}(x)$, where $\beta_{j,k,m}$ are smooth functions on $\mathbb R$.

%\begin{As}\label{H3}
%$\mc{V}\neq\emptyset$ and $\mathcal{V}\cap\{\xi=0\}=\emptyset.$ Moreover, the contact order of $\Gamma_1(E_0)$ and $\Gamma_2(E_0)$ is finite at each crossing point, i.e., $m_\rho<\infty$ for all $\rho\in \mathcal{V}$. 
%\end{As}
 
%Notice that under Assumption~\ref{H3},  the number of crossing points is finite since $\Gamma_1(E_0)$ is compact.

\subsection{Main result}
In the following, we use the same letter $U$ for the principal symbol of the operator $U$, that is, $U(x,\xi):= r_0(x)+ir_1(x) \xi$, $(x,\xi)\in T^* \mathbb R$. Set
\ben\label{M0}
m_0:=\max \left\{m_\rho; \rho\in \mathcal{V}  \right\}.
\een
%and, for an integer $m$, 
%\ben
%\mathcal{V}^{(m)}:=\{\rho\in\mathcal{V} ;\,m_\rho=m\}.
%\een
%In particular, $\mathcal{V}^{(m_0)}$ is the set of the crossing points with maximal contact order. 
Our main result is the following.
\begin{theorem}\label{MAINTH}
Let Assumptions \ref{H1}, \ref{H2} and \ref{H3} hold, and let $\mathfrak{B}_h$ and $\mathcal R$ be given by \eqref{DEFFF} and \eqref{DEFR} with $\delta_1=\delta_2=Lh$ for an arbitrarily fixed $L>0$. Then there exist a bijective map 
\ben
z_h : \mathfrak{B}_h \to {\rm Res} (P)\cap {\mathcal R}
\een 
and a non-negative smooth function $\mathcal{D}(E)=\mathcal{D}(E;h)$ of $E$ near $E_0$ uniformly bounded with respect to $h>0$ small enough (given explicitly by \eqref{eq:Conclusion}) such that for any $E\in \mathfrak{B}_h$ one has 
\begin{equation}\label{Realpart}
\vert z_h(E) - E \vert =  \mathcal{O}(h^{\frac{m_0+3}{m_0+1}}),
\end{equation}
\begin{equation}\label{Impart}
{\rm Im}\, z_h(E) = -\mathcal{D}(E)h^{\frac{m_0+3}{m_0+1}}+\ord(h^{\frac{m_0+3}{m_0+1}+\kappa}),
\end{equation}
uniformly as $h\to0^+$, where $\kappa=\min\left\{1/3,\,1/(m_0+1)\right\}$.
\end{theorem}

The leading coefficient $\mathcal{D}(E)$ in the asymptotics \eqref{Impart} may vanish at some energies $E\in \mathfrak{B}_h$ for each small $h$.
This means that, in such a case, there exist resonances closer to the real axis than those  with imaginary part of order $h^{(m_0+3)/(m_0+1)}$.
This phenomenon of double lines of resonances occurs due to the existence of a
 \textit{directed cycle} of the characteristic set $\Gamma(E_0)$ 
 created by crossings. We illustrate this in the next paragraph on a simple model. We conclude this section with the definition of
 directed cycle.
 
\begin{definition}
\label{GCT}
We call a curve $\gamma: [0,T]\to \Gamma$ a \textit{generalized classical trajectory} if it is parametrized by the time variable of Hamiltonian flow on each $\Gamma_j$ $(j=1,2)$.
In particular, we call it {\it directed cycle} if it is a simple closed curve.
\end{definition}
In other words, a generalized classical trajectory is a usual classical trajectory of $p_1$ on $\Gamma_1\setminus \Gamma_2$ and of $p_2$ on $\Gamma_2\setminus \Gamma_1$, but it is allowed to change from one to the other at a crossing point. 
One may find an analogous terminology in \cite{W}. 
In \cite{Hi2}, one of the authors showed that a directed cycle made by two non-trapping classical trajectories creates  resonances with width of order $h\log(1/h)$.

\subsection{Example}\label{subSec:App}
We restate the above result with the explicit value of ${\mathcal D}(E)$ for a simplest model including the one studied in \cite{FMW3} which corresponds to the case $m_0=1$.  

Suppose that the set $\{V_1=V_2\leq E_0\}$ consists of only one point. Without loss of generality we assume that 
\ben
E_0>0, \quad  \quad \{x\in\R;\,V_1(x)=V_2(x)\le E_0\}=\{0\}  \quad \text{and} \quad V_1(0) = V_2(0) = 0.
\een
Suppose moreover that $V_1$ satisfies the simple well condition given by Assumption \ref{H2} (i) with $a_0<0<b_0$, and there exists $c_0\in\R$ such that (see Fig. \ref{Fig:Ex1})
\ben
\frac{V_2(x)-E_0}{x-c_0}>0, \quad \forall x\in \mathbb R.
\een
Under these conditions, the set of crossing points $\mathcal{V}$ in the phase space consists of two points $\rho_\pm = (0,\pm\sqrt{E_0})$, and there exists a  unique directed cycle $\gamma\subset \Gamma(E_0)$ other than $\Gamma_1(E_0)$. The associated action integral defined by 
\ben
S_\gamma(E):=\int_\gamma\xi dx,
\een
can be explicitly written as 
\be\label{eq:SforEx}
S_\gamma(E)=2\int_{a(E)}^0\sqrt{E-V_1(x)}\,dx+2\int_0^{c(E)}\sqrt{E-V_2(x)}\,dx,
\ee
where $a(E)$ denotes the unique solution to $V_1(x)=E$ near $a_0$ and $c(E)$ denotes the unique solution to $V_2(x) = E$ near $c_0$ (in particular, $a(E_0)=a_0$ and $c(E_0)=c_0$). Finally assume that 
\ben
m_0 := \min\{k\in\N; V_1^{(k)}(0)\ne V_2^{(k)}(0)\} <\infty,
\een
and set $v_0:= |V_2^{(m_0)}(0)-V_1^{(m_0)}(0)|$.  

\begin{figure}[t]
\centering
\includegraphics[bb=0 0 1376 859, width=10cm]{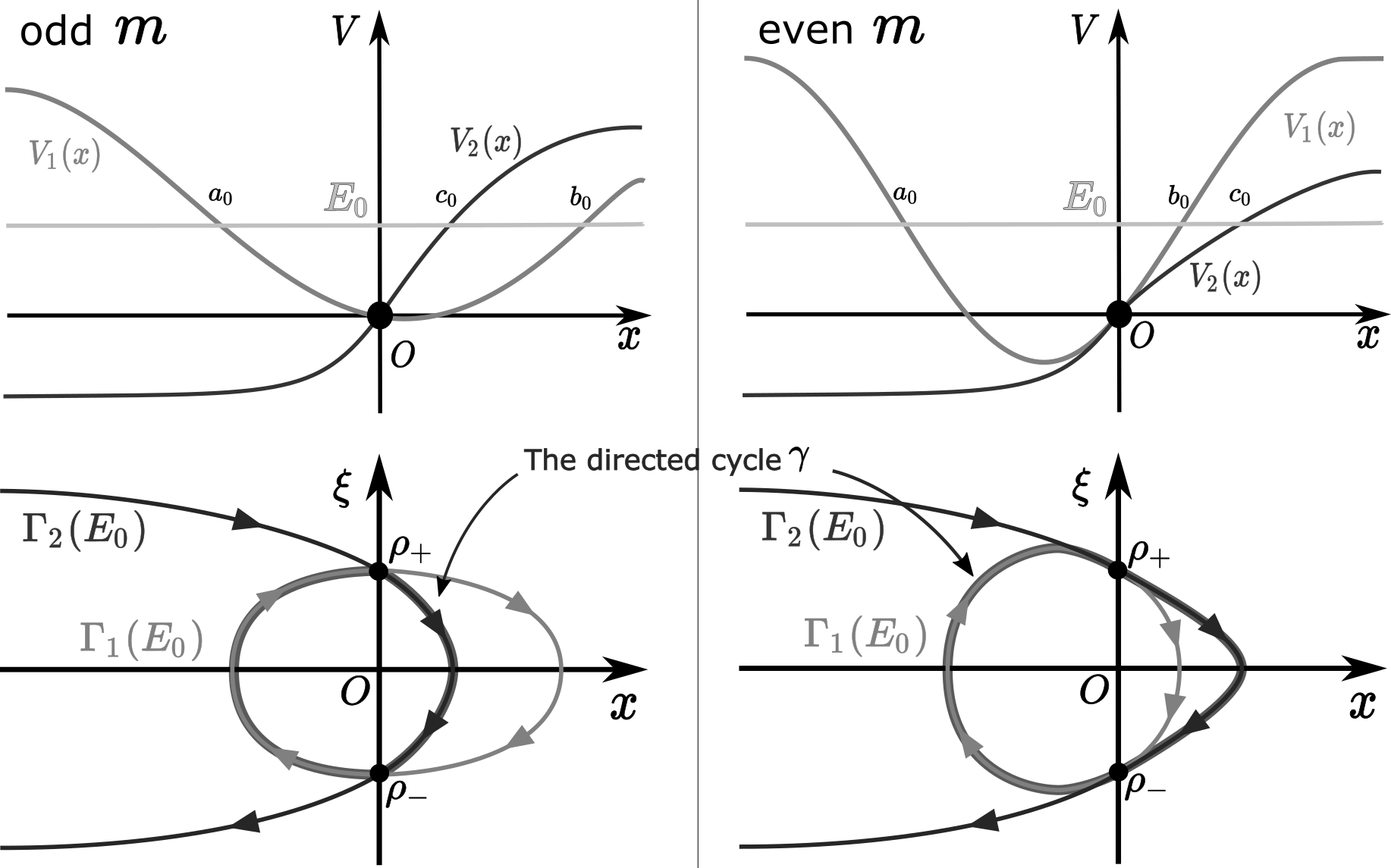}
\caption{Typical model}
\label{Fig:Ex1}
\end{figure}

Then, with the notations of Theorem \ref{MAINTH}, the leading coefficient of the width of the resonances of $P$ is given by
\ben
\mathcal{D}(E)
=\frac{2E^{-\frac{m_0}{m_0+1}}}{\cA'(E)}\left[\im\left(
\eta_{m_0}  \overline{U(\rho_+)}
e^{iS_\gamma(E)/2h}\right)\right]^2,
\een
where $\eta_{m_0}\neq 0$ is given by
\begin{align}\label{eq:eta-m}
\eta_{m_0}=
\left\{
\begin{aligned}
&\exp \left(i\frac{\pi}{2(m_0+1)}\right)
\bm{\Gamma}\left(\frac{m_0+2}{m_0+1}\right)\left(\frac{2(m_0+1)!}{v_0}\right)^{\frac{1}{m_0+1}}  && \text{if $m_0$ is odd},\\
&\cos\left(\frac{\pi}{2(m_0+1)}\right)
\bm{\Gamma}\left(\frac{m_0+2}{m_0+1}\right)\left(\frac{2(m_0+1)!}{v_0}\right)^{\frac{1}{m_0+1}} && \text{if $m_0$ is even}.
\end{aligned}
\right. 
\end{align}
Here $\bm{\Gamma}$ is the usual gamma function defined by 
\begin{equation}\label{Gammaf}
\bm{\Gamma}(z) := \int_0^{+\infty} t^{z-1} e^{-t} dt, \quad {\rm Re}\,z>0.
\end{equation}
Notice that even when $U(\rho_+)\neq 0$, the coefficient $\mathcal{D}(E)$ vanishes on a discrete subset determined by the action $S_\gamma(E)$ associated to the directed cycle $\gamma$. This set is given by the energies $E$ such that 
%$\mathcal{D}(E)$ may vanish even when $U(\rho_+)\neq 0$. 
%For instance, %in the case $m_0$ even, 
%$\mathcal{D}(E)$ vanishes at energies $E$ such that 
\ben
\left\{\begin{aligned}
&h^{-1}S_\gamma(E)-\pi+2\tan^{-1}\left( \frac{r_0(0)}{r_1(0) \sqrt{E_0}}\right) + \frac\pi{(m_0+1)} && \text{if $m_0$ is odd},\\
&h^{-1}S_\gamma(E)-\pi+ 2\tan^{-1}\left( \frac{r_0(0)}{r_1(0) \sqrt{E_0}}\right) && \text{if $m_0$ is even},
\end{aligned}\right.
\een
belongs to $2\pi\Z$. 
We here set $\tan^{-1}\left(\frac{r_0(0)}{r_1(0)\sqrt{E_0}}\right)=\frac \pi2$ when $r_1(0)=0$. 

%==============	$General results$	===============
\section{Microlocal connection at a crossing point} \label{SEC3}

%\subsection{Semiclassical and microlocal terminologies}
\subsection{Definition and asymptotics for the microlocal transfer matrix
%Preliminary facts and statement of the results
}
We recall briefly some basic notions of semiclassical and microlocal analysis, referring to the textbooks \cite{DiSj, Ma, Zw} for more details.
For $N\in\N$ and $m\in \R$, we introduce the class $S_N^m$ of matrix-valued symbols
\ben
S_N^m:=\left\{a\in C^\infty(\R^2;\C^{N\times N})\,; \, 
\Vert \p_x^k \p_\xi^l a(x,\xi) \Vert_{N\times N} \le C_{k,l}(1+\xi^2)^{m/2},\ %(1+\xi^2)^{(m-l)/2},
\forall k,l\in\N\right\}.
\een
For a symbol $a\in S_N^m$, the corresponding $h$-pseudodifferential operator, denoted $a^w(x,hD_x)$, is defined as an unbounded operator in $L^2(\R;\C^N)$ %on the Sobolev space $H^m(\R)$ 
via the $h$-Weyl quantization
\begin{align}\label{WeylQ}
 a^w(x,hD_x)u(x)
:=\frac{1}{2\pi h}\int_{\R^2}e^{i(x-y)\xi/h}a\left(\frac{x+y}{2},\xi\right)u(y)dyd\xi.
\end{align}
This is a bounded operator in $L^2(\mathbb R; \mathbb C^N)$ if $m\le 0$ according to the Calder\'on-Vaillancourt Theorem, see e.g. \cite[Theorem 7.11]{DiSj}. 

Let $(x_0,\xi_0)\in\R^2$ and $f=f(x;h)\in L^2(\R; \mathbb C^N)$ with $\|f\|_{L^2}\le 1$. 
We say that $f$ is microlocally $0$ near $(x_0,\xi_0)$ and we write
\ben
f(x;h)\equiv0\qtext{near}(x_0,\xi_0),
\een
if there exists a symbol $\chi=\chi(x,\xi)\in S_N^0$ with $\det\, \chi(x_0,\xi_0)\neq0$ such that 
\ben
\| \chi^w (x,hD_x)f\|_{L^2(\R)}=\ord(h^\infty).
\een 
We also say that $f$ is microlocally $0$ near a set $\Omega\subset\R^2$ if it is microlocally $0$ near each point of $\Omega$. 
We define $\ope{WF}_h(f)$, the semiclassical wave front set of $f$, as the set of all points of $\R^2$ where $f$ is not microlocally $0$. 
The following results for scalar pseudo-differential operators are well-known (see e.g. \cite[Theorem 12.5]{Zw}).  Let $a\in S^0_1$. If $f$ satisfies $a^w(x,hD)f\equiv0$ microlocally in $U$, we have $\ope{WF}_h(f)\cap U\subset \{a=0\}$. 
Moreover, if $\p a\neq0$ on $\{a=0\}$, the semiclassical wave front set $\ope{WF}_h(f)$ is invariant under the Hamiltonian flow of $a$. 

%\subsection{Microlocal connection formulae}\label{sub:connection}

Throughout this section, we fix $E_0\in\R$ as in the previous section and consider $E$ belonging to our reference box $\cR$ defined by \eqref{DEFR} with $\delta_1=\delta_2=Lh$ for an arbitrarily fixed $L>0$. We set $\Gamma=\Gamma(E_0):=\Gamma_1(E_0)\cup\Gamma_2(E_0)$. 
%Let us start with the following preliminary result.

The following proposition is a simple consequence of a general fact
that the equation $a^w(x,hD_x)u=0$ for a matrix-valued symbol $a(x,\xi)\in S_N^m$ reduces microlocally near $(x_0,\xi_0)\in \mathbb R^2$ to a scalar equation
if $0$ is a simple eigenvalue of $a(x_0,\xi_0)$.

\begin{proposition}\label{prop:RedScal}
The space of microlocal solutions to the system $(P-E)w=0$ near each connected subset of $\Gamma\setminus\mc{V}$ is one dimensional.
\end{proposition}

%\begin{proof}
%This result is in fact a simple consequence of a well known general result stating that given a pseudodifferential system $q^w(x,hD_x)$ with matrix valued symbol $q=q(x,\xi)$, if $0$ is a simple eigenvalue of $q(x_0,\xi_0)$ for some $(x_0,\xi_0)\in T^*\mathbb R$, then the equation $q^w(x,hD_x)u=0$ reduces microlocally near $(x_0,\xi_0)$ to a scalar equation.
%\end{proof}

Microlocally near each point %$\edg_j^\dir$ $(j=1,2$, $\dir=\out,\inc)$, 
of $\Gamma\setminus\mc{V}$, except at the turning points $\Gamma\cap\{\xi=0\}$, we have the following WKB microlocal solution to the system $(P-E)w=0$ constructed in \cite[Proposition~5.4]{FMW3}. Here and in what follows, $\hat{j}$ stands for the complement of $j$ in the set $\{1,2\}$.
\begin{proposition}{\cite[Proposition~5.4]{FMW3}} \label{prop:mlWKB}
Let $j\in\{1,2\}$. 
For any $\rho_j=(x_j,\xi_j)\in\Gamma_j(E_0)\setminus\mc{V}$, $\xi_j\neq 0$, there exists a microlocal solution $f_{\rho_j}$ to the equation
\be
(P-E)w\equiv0\qtext{ near}\rho_j,
\ee
of the form
\begin{equation}
f_{\rho_j}(x,h;x^*)= e^{(\ope{sgn}\xi_j)i\phi_j(x;x^*)/h}\begin{pmatrix}
\sigma_{j,1}(x,h)\\
\sigma_{j,2}(x,h)
\end{pmatrix},
\end{equation}
where the phase function $\phi_j$ normalized at some point $x^*$ in the same connected component of the classically allowed region of $P_j-E$ as $x_j$ is given by
\be\label{eq:N-Phase}
\phi_j(x;x^*):=
\int_{x^*}^x\sqrt{E_0-V_j(y)}\,dy,
\ee
and the symbols $\sigma_{j,k}(x,h)$ are smooth functions of $x$ with asymptotic expansion in powers of $h$ of the form $\sigma_{j,k}(x,h)\sim\sum_{l\ge0}h^l\sigma_{j,k,l}(x)$ such that
\ben
\sigma_{j,j,0}(x)=(E_0-V_j(x))^{-1/4}, \quad\sigma_{j,\hat j,0}(x)=0\quad \text{and} \quad \sigma_{j,j,l}(x_j)=0\quad\text{for all } l\ge1.
\een
\end{proposition}

Now let $\rho=(\rho_x,\rho_\xi)\in\mc{V}$ be a crossing point. Recall from Assumption \ref{H3} that $\rho_\xi\neq 0$. 
We study the behavior of microlocal solutions to $(P-E)w\equiv0$ near an $h$-independent small neighborhood of $\rho$. In this neighborhood, each $\Gamma_j(E_0)$, $j=1,2$, has an incoming classical trajectory $\edg_j^\flat$ and outgoing one $\edg_j^\sharp$ consisting of points $\exp tH_{p_j}(\rho)$ with negative and positive small $t$.
We take a point $\rho_j^\flat$ on $\edg_j^\flat$ and a point $\rho_j^\sharp$ on $\edg_j^\sharp$.

Let $f_j^\dir:=f_{\rho_j^\dir}(x,h;\rho_x)$ be the WKB microlocal solution to the system near $\rho_j^\dir$ given by the previous Proposition, normalized at $x^*=\rho_x$. If $w\in L^2(\R;\C^2)$ satisfies
\begin{equation}
(P-E)w\equiv 0\qtext{ near}\rho,
\end{equation}
then according to Proposition~\ref{prop:RedScal},  there exist constants $\alpha_j^\dir=\alpha_j^\dir(w,E,h)$ $(j=1,2$, $\dir=\out,\inc)$ such that
\begin{equation}\label{eq:alpha-w}
w\equiv \alpha_j^\dir f_j^\dir\qtext{ near} \rho_j^\dir.
\end{equation}
Let $m=m_\rho$ be the contact order at $\rho$ of $\Gamma_1(E_0)$ and $\Gamma_2(E_0)$ defined by \eqref{CCOO}, which is finite according to Assumption \ref{H3}.

The following two theorems are the main results of this section. The first one ensures the well-definedness of the transfer matrix at the crossing point $\rho$, and the second one gives the semiclassical asymptotics of this matrix.

\begin{theorem}\label{thm:T-exists}
The space of microlocal solutions to the equation $(P-E)w\equiv0$ near $\rho$ is two dimensional. More precisely, there exists a matrix $T_\rho=T_\rho(E,h;\rho_1^\inc,\rho_2^\inc,\rho_1^\out,\rho_2^\out)$  such that for any microlocal solution $w\in L^2(\R;\C^2)$ near $\rho$ with $\|w\|_{L^2}\le1$, one has
\be\label{transfer}
\biggl(\begin{matrix}
\alpha_1^\out\\
\alpha_2^\out
\end{matrix}
\biggr)
=
T_\rho(E,h)
\biggl(
\begin{matrix}
\alpha_1^\flat\\
\alpha_2^\flat
\end{matrix}\biggr).
\ee
\end{theorem}

%The following Theorem gives the semiclassical asymptotics of the matrix $T_\rho$.

\begin{theorem}\label{thm:ConnectionAllowed}
Let $L>0$ be an $h$-independent constant. For $E$ in $\cR_{E_0}(Lh,Lh)$, one has 
\begin{equation}\label{MAAS}
T_\rho(E,h)={\rm Id}_2+h^{\frac 1{m+1}}T_{\rm sub} +\ord(h^{\frac 2{m+1}}),
\end{equation}
uniformly with respect to $h>0$ small enough, where ${\rm Id}_2$ is the  $2\times 2$ identity matrix and 
\ben
 T_{\rm sub}=-i\begin{pmatrix}
0& \overline{\omega_\rho}\\
\omega_\rho&0
\end{pmatrix}, 
\een
where $\omega_\rho \in \mathbb C$ is a constant independent of $\rho_j^\dir$ $(j=1,2, \dir=\flat, \sharp)$ explicitly given by \eqref{omega} (see also Remark~\ref{rem:gencase}). 
\end{theorem} 

The rest of this section is devoted to the proofs of these results. 

\subsection{Outline of the proof of the theorems}
We first remark that if $\rho=(\rho_x,\rho_\xi)$ is a crossing point with $\rho_\xi\ne 0$, then $\breve\rho=(\rho_x,-\rho_\xi)$ is also a crossing point.

%To prove Theorem \ref{thm:T-exists}, the first step is the following proposition  which brings the study of microlocal solutions near $\rho=(\rho_x,\rho_\xi)\in \mathcal{V}$ to that of local exact solutions to the system $(P-E)w=0$. 

Theorem \ref{thm:T-exists} follows from the two propositions below. The first one brings the study of microlocal solutions near $\rho=(\rho_x,\rho_\xi)\in \mathcal{V}$ to that of local exact solutions to the system $(P-E)w=0$, and the second one shows that the space of exact solutions supported microlocally near the crossing point is two dimensional. 
\begin{proposition}\label{lem:ApproByEx}
Let $I$ be a real interval containing $\rho_x$ in its interior. For any microlocal solution $w\in L^2(I;\C^2)$ to the system $(P-E)w=0$ near $\rho$ with $\|w\|_{L^2(I)}\le1$, there exists an exact solution $\til{w}$ such that
\be
w\equiv\til{w}\qtext{ near}\rho.
\ee
\end{proposition}

\begin{proposition}\label{prop:Basis-pm}
There exists a basis $(w_1^+,w_2^+,w_1^-,w_2^-)$ of exact solutions on $I$ to $(P-E)w=0$ such that 
\be\label{eq:Basis1234}
\|w_j^\pm\|_{L^2(I)}=1 \qtext{and}
w_j^-\equiv 0\qtext{ near}\rho, \quad j=1,2.
\ee
\end{proposition}
We prove these propositions in Section~\ref{section3.4}. 

\begin{comment}
Taking this result into account, we construct locally on $I$ a basis $ (w_1^+,w_2^+,w_1^-,w_2^-)$ of exact solutions to $(P-E)w=0$ such that 
\be\label{eq:Basis1234}
\|w_j^\pm\|_{L^2(I)}=1 \qtext{and}
w_j^-\equiv 0\qtext{ near}\rho, \quad j=1,2.
\ee
This implies Theorem~\ref{thm:T-exists}. 
\end{comment}

We briefly explain the scheme of the construction of these solutions due to \cite{FMW1}. In the following, we use the notations $U_1=U$ and $U_2=U^*$. 
\begin{comment}
Setting $w=\begin{pmatrix}v_1\\v_2\end{pmatrix}$, the system $(P-E)w=0$ is equivalent to 
\be\label{eq:System}
\left\{
\begin{aligned}
(P_1-E)v_1=-hU_1v_2,\\
(P_2-E)v_2=-hU_2v_1.
\end{aligned}\right.
\ee
\end{comment}

Fix a compact interval $I$ containing $\rho_x$. 
For each $j=1,2$, let $(u_j, \breve{u}_j)$  be a pair of two linearly independent solutions to the scalar homogeneous equation $(P_j-E)u=0$ on $I$.
We define an integral operator $K_j=K(\breve{u}_j,u_j,s_j,t_j)$ on $C(I)$ depending on $s_j,t_j\in I$ by
\be\label{eq:Def-K-Gen}
K_j[f](x):=\frac{1}{h^2\W(\breve{u}_j,u_j)}\left(\breve{u}_j(x)\int_{s_j}^x u_j(y)f(y)dy-u_j(x)\int_{t_j}^x \breve{u}_j(y)f(y)dy\right).
\ee
Here $\W(\breve{u}_j,u_j)=\det(\breve{u}_j,u_j)$ stands for the Wronskian of $\breve{u}_j$ and $u_j$.

\begin{comment}
For $j=1,2$, let $\til{v}_j$ be a local solution on $I$ to the homogeneous equation $(P_j-E)v=0$. Then a local solution on $I$ to the system \eqref{eq:System} is given by solving the integral equation
\be\label{eq:IntSystem}
\left\{
\begin{aligned}
v_1=\til{v}_1-hK_1U_1v_2,\\
v_2=\til{v}_2-hK_2U_2v_1.
\end{aligned}\right.
\ee
This integral system can be solved 
\end{comment}
If at least one of the infinite sums $J_1=J_1(K_1,K_2)$ or $J_2=J_2(K_1,K_2)$ defined by
\be\label{eq:Inf-Sum}
J_1(K_1,K_2):=\sum_{k\ge0}(\til{K}_1\til{K}_2)^k,\quad
J_2(K_1,K_2):=\sum_{k\ge0}(\til{K}_2\til{K}_1)^k
\qtext{with}\til{K}_j:=-hK_jU_j
\ee 
converges in some suitable space,
then for any solution $v_1$ to $(P_1-E)v=0$ or for any solution $v_2$ to $(P_2-E)v=0$, the vector-valued function $w_1$ or $w_2$ defined by
\be\label{Def-Sol-Gen}
w_1(K_1,K_2,v_1):=
\begin{pmatrix}J_1v_1\\\til{K}_2J_1v_1\end{pmatrix},\quad
w_2(K_1,K_2,v_2):=
\begin{pmatrix}\til{K}_1J_2v_2\\J_2v_2\end{pmatrix},
\ee
is an exact solution to the system $(P-E)w=0$. 

Suppose $\rho_\xi>0$. This means that the incoming trajectories are on the left of the crossing point $\rho$ and the outgoing ones on the right.
We take $u_j$, $\breve u_j$, $j=1,2$ as in Proposition \ref{prop:SclAllowed} below, i.e. $u_j$ and  $\breve u_j$ are WKB solutions to $(P_j-E)u=0$ near $\rho$ and $\breve \rho$ respectively. If we take 
$s_j$, $t_j$ both on the left of $\rho_x$, the exact solutions $w_1^\inc=w_1(K_1,K_2,u_1)$, $w_2^\inc=w_2(K_1,K_2,u_2)$ defined by \eqref{Def-Sol-Gen} satisfies
%Now, to prove the asymptotic formula for $T_\rho$ given by Theorem~\ref{thm:ConnectionAllowed}, we first construct two solutions $w_1^\inc$, $w_2^\inc$ to the equation  such that 
\be\label{eq:BehaveInc}
w_1^\inc\equiv 
\left\{
\begin{aligned}
&f_1^\inc +\mathcal{O}(h)&& \text{ near} \; \rho_1^\inc,\\
&\mathcal{O}(h) && \text{ near}\;  \rho_2^\inc,
\end{aligned}\right.
\qquad
w_2^\inc\equiv 
\left\{
\begin{aligned}
&\mathcal{O}(h)&& \text{ near} \; \rho_1^\inc,\\
&f_2^\inc +\mathcal{O}(h) && \text{ near}\;  \rho_2^\inc.
\end{aligned}\right.
\ee
The main work consists in the computation of the microlocal asymptotic behavior of these solutions on the outgoing trajectories: 
\be\label{eq:BehaveOut}
w_1^\inc\equiv 
\left\{
\begin{aligned}
&t_{1,1}^\out f_1^\out +\mathcal{O}(h)&& \text{ near} \; \rho_1^\out,\\
&t_{2,1}^\out f_2^\out +\mathcal{O}(h) && \text{ near}\;  \rho_2^\out,
\end{aligned}\right.
\qquad
w_2^\inc\equiv 
\left\{
\begin{aligned}
&t_{1,2}^\out f_1^\out +\mathcal{O}(h)&& \text{ near} \; \rho_1^\out,\\
&t_{2,2}^\out f_2^\out +\mathcal{O}(h) && \text{ near}\;  \rho_2^\out.
\end{aligned}\right.
\ee
Theorem~\ref{thm:T-exists}  implies that the coefficients $t_{j,k}^\out$ give the transfer matrix modulo ${\mathcal O}(h)$:
\be
T_\rho=T^\out+\ord(h),\quad T^\out=(t_{j,k}^\out)_{1\leq j,k\leq 2}.
\ee

For the computation of \eqref{eq:BehaveOut} (as well as for the convergence of the infinite series $J_1$ and $J_2$), we apply the following degenerate stationary phase estimate in one dimension (see for instance \cite{Ho1}) for the oscillatory integral
\be
\label{Ihint}
\mc{I}(h):=\int_I \sigma(x)\exp\left(\frac ih \phi(x)\right)dx
\ee
with a smooth real valued phase $\phi(x)$ and a smooth symbol $\sigma(x,h)$ uniformly bounded in $h$ with its derivative.
In the application, $\phi$ will be in the form
\be
\label{+-+-}
\phi=\pm\int^x\left(\sqrt{E-V_1(y)}\pm\sqrt{E-V_2(y)}\right)dy.
\ee
The critical points of the phase function $\phi$ exist only when the sign between the two square roots is $-$, and, in that case, they are  the crossing points of the potentials $V_1$ and $V_2$.

\begin{lemma}\label{lem:OI-degm}
Let $I\subset \mathbb R$ be an interval, and suppose that $\phi(x)$
and $\sigma(x,h)$ are smooth functions in a neighborhood of $\bar I$. 
Suppose moreover that $\phi$ is real-valued, and $\sigma,$ $\p_x\sigma$ are bounded as $h\to0^+$.
\begin{enumerate}
\item
If $\phi'(x)$ never vanishes in $\bar I$, then we have $\mc{I}(h)=\ord (h)$.

\item If $x_0\in \abxring{I}$ is  the unique stationary point of $\phi$ in $\bar{I}$, and if $\phi'(x_0)=\cdots =\phi^{(m)}(x_0)=0$, $\phi^{(m+1)}(x_0)\ne 0$, then we have
$$
\mc{I}(h) = 
\til{\omega}e^{\frac ih\phi(x_0)} h^{\frac 1{m+1}}+\ord(h^{\frac 2{m+1}}),
$$
where
\be
\til{\omega}=\mu_m\sigma(x_0)\left(\frac{(m+1)!}{\left|\phi^{(m+1)}(x_0)\right|}\right)^{\frac 1{m+1}} \bm{\Gamma}\left(\frac{m+2}{m+1}\right),
\ee
\begin{equation*}
\mu_m=\left\{
\begin{aligned}
&e^{\frac{i\pi}{2(m+1)}\ope{sgn}\phi^{(m+1)}(x_0)}	& & \qtext{when $m$ is odd,}\\
&\cos\left(\frac{\pi}{2(m+1)}\right)	& & \qtext{when $m$ is even.}
\end{aligned}
\right.
\end{equation*}
\end{enumerate}
Here $\bm{\Gamma}$ is the usual Gamma function defined by \eqref{Gammaf}.
\end{lemma}

%We use it to show the convergence of the infinite series $J_1$ and $J_2$ given by \eqref{eq:Inf-Sum} (Proposition~\ref{prop:ConvSer}), to observe microlocal behavior of the solutions (Proposition~\ref{prop:MLbehavior}), and to compute the linear relationship between two bases of solutions (Lemma~\ref{lem:ComputeAlpha}).

\subsection{Asymptotics of the transfer matrix at a crossing point (Proof of Theorem~\ref{thm:ConnectionAllowed})}
%\subsection{Proof of Theorem~\ref{thm:ConnectionAllowed}}
%{Proof of Theorem~\ref{thm:ConnectionAllowed}}

In this paragraph, we suppose that Theorem \ref{thm:T-exists} holds, that is, the transfer matrix $T_\rho$ is well defined, and we prove Theorem \ref{thm:ConnectionAllowed}. The proof of Theorem \ref{thm:T-exists} will be given in the next paragraph.
%===============	Figure  Connection	==================

We assume without loss of generality that $\rho_x=0$ and $V_1(0)=V_2(0)=0$.  Assumption \ref{H3} then implies $E_0>0$.
Let $I$ be an $h$-independent small interval containing $0$ inside such that $0$ is the only zero of $V_1$ and $V_2$ in $I$. We restrict our study to the region $I\times \R_\xi$ and set 
$\Gamma^I=\Gamma^I(E_0):=\Gamma(E_0)\cap (I\times \R_\xi)$,
$\Gamma_j^I(E_0):=\Gamma_j(E_0)\cap (I\times \R_\xi)$.

The two crossing points $\rho=(0,\sqrt{E_0})$ and $\breve\rho=(0,-\sqrt{E_0})$ divide the characteristic set $\Gamma^I$ into eight curves  (see Fig. \ref{Fig:A}): For $j=1,2$, 
\begin{align*}
&\edg_j^\out:=\Gamma_j^I\cap\{x>0, \xi>0\}, \quad
\breve\edg_j^\out:=\Gamma_j^I\cap\{x<0, \xi<0\},\\ 
&\edg_j^\inc:=\Gamma_j^I\cap\{x<0, \xi>0\},\quad 
\breve\edg_j^\inc:=\Gamma_j^I\cap\{x>0, \xi<0\}.
\end{align*}

\begin{figure}[t]
\centering
\includegraphics[bb=0 0 760 389, width=11.5cm]{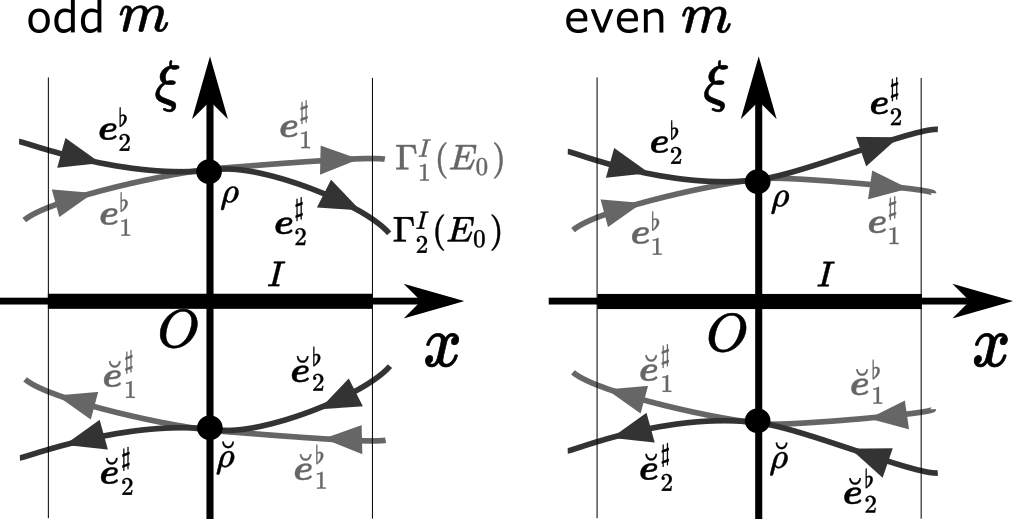}
\caption{The characteristic set near a pair of crossing points}
\label{Fig:A}
\end{figure}

We use bases of solutions to the scalar equation $(P_j-E)u=0$ $(j=1,2)$ to construct solutions to the system $(P-E)w=0$.  The scalar equations can be solved locally by the method of successive approximation (see \cite{Ol}) since $E_0-V_j$ is $h$-independently positive on $I$. 
\begin{proposition}\label{prop:SclAllowed}
For $j=1,2$, there exists a pair of solutions $(\breve{u}_j, u_j)=(\breve{u}_j(x,h,E), u_j(x,h,E))$  to the scalar equation $(P_j-E)u=0$  satisfying the following properties:
\begin{enumerate}
\item $\breve{u}_j$ and $u_j$ admit  the following asymptotic behaviors
\be\label{eq:scalarWKB-allowed}
\breve{u}_j=\frac{1+\ord(h)}{\sqrt[4]{E-V_j(x)}} e^{-i\int_0^x\sqrt{E-V_j(t)}\,dt/h},\quad
u_j=\frac{1+\ord(h)}{\sqrt[4]{E-V_j(x)}} e^{i\int_0^x\sqrt{E-V_j(t)}\,dt/h},
\ee
as $h\to0^+$ uniformly on $I$. 
\item For real $E$ near $E_0$, one has $\breve{u}_j=\overline{u_j}$, and 
\be
\ope{WF}_h(\breve{u}_j)\subset\Gamma_j^I(E_0)\cap\{\xi<0\},\quad
\ope{WF}_h(u_j)\subset\Gamma_j^I(E_0)\cap\{\xi>0\}.
\ee
\item They are normalized to satisfy
\be
\W(\breve{u}_j,u_j)=\breve{u}_j(u_j)'-(\breve{u}_j)'u_j=2ih^{-1}.
\ee
\end{enumerate}
\end{proposition}
\begin{remark}
Only the first property \eqref{eq:scalarWKB-allowed} of the three previous properties is essential for the existence. 
Once one obtains a pair of solutions which satisfies \eqref{eq:scalarWKB-allowed}, the pair $(\breve{u}_j,u_j)$ satisfying the other two properties is obtained by a suitable linear combination. 
\end{remark}

We denote by $C(I)$ the Banach space of continuous functions on $I$ endowed with the norm $\left\|f\right\|_{C(I)}=\sup_{x\in I}\left|f(x)\right|$. 
For $j=1,2$ we take $s_j$, $t_j\in I$ and define operators by 
\be
\label{kj}
K_j:=K(\breve u_j,u_j,s_j,t_j):C(I)\to C(I)
\ee
 \begin{proposition}\label{prop:ConvSer} 
The operator $\tilde K_j=-hK_jU_j$
 is uniformly bounded for $E\in\cR$ and $h$ small enough. 
Moreover, we have the following uniform estimates 
\be\label{eq:ResEstim}
\left\|\til{K}_1\til{K}_2\right\|_{\mc{B}(C(I))}=\ord(h^{1/(m+1)}),\quad
\left\|\til{K}_2\til{K}_1\right\|_{\mc{B}(C(I))}=\ord(h^{1/(m+1)}).
\ee 
\end{proposition}

\begin{proof} 
We follow the scheme of \cite[Proposition~3.1]{FMW3}, where the authors estimated the norm of similar operators for $m=1$. Remark that the estimate $h^{1/2}$ here in \eqref{eq:ResEstim} is better than that of $h^{1/3}$ in \cite[Proposition~3.1]{FMW3}. This is because the latter is a global estimate on $(-\infty, 0)$ which contains turning points.

For simplicity, we assume that the interaction operator is a multiplication operator $r_0(x)$. The general case of first differential operator can be treated similarly after an integration by parts as in \cite{FMW3}. 

The operators $\til{K}_j$ are obviously bounded in $C(I)$ since
\begin{align*}
\til{K}_jf(x)&=-hK_jU_j f(x)
=\frac i2\left(\breve u_j(x)\int_{s_j}^xr_0(y)u_j(y)f(y)dy-u_j(x)\int_{t_j}^xr_0(y)\breve u_j(y)f(y)dy\right)
%&=\frac i2\Biggl[\breve u_j(x)\left(h[(-1)^{j-1}iu_j(y)r_1(y)f(y)]_{y=s_j}^x+\int_{s_j}^x({}^tU_ju_j)(y)f(y)dy\right)
%\\&\qquad-u_j(x)\left(h[(-1)^{j-1}i\breve u_j(y)r_1(y)f(y)]_{y=t_j}^x+\int_{t_j}^x({}^tU_j\breve u_j)(y) f(y)dy\right)\Biggr],
\end{align*}
%where the transpose ${}^tU_j$ of $U_j$ is a first order $h$-differential operator with bounded coefficients. 
and $r_0$, $u_j$ and $\breve u_j$ are of $\ord (1)$ on $I$.
%$hD_x\breve u_j$ and $hD_xu_j$ 

As for the estimate for the composed operator $\til{K}_1\til{K}_2$, we study the integral of the form
\be\label{eq:CompResolv}
\int_{s}^xr_0(y)v_1(y)v_2(y)\left (\int_{s'}^yr_0(t)  v_3(t)f(t)dt\right )dy
\ee
for $s,s'\in I$ and $(v_1,v_2,v_3)=(u_1,\breve u_2,u_2), (u_1,u_2,\breve u_2), (\breve u_1,\breve u_2,u_2)$ and $(\breve u_1,u_2,\breve u_2)$. 
If we write
\eqref{eq:scalarWKB-allowed} as
$\breve{u}_j=\breve a_j(x,h) e^{-i\phi_j(x)/h},\quad
{u}_j=a_j(x,h) e^{i\phi_j(x)/h}$, $\phi_j(x)=\int_0^x\sqrt{E-V_j(t)}\,dt$, then
\eqref{eq:CompResolv} can be written in the form
\be
\int_s^x \sigma(y,h)e^{i\phi(y)/h}dy,
\ee
where 
$$
\sigma(y,h)=r_0(y)b_1(y,h)b_2(y,h)\int_{s'}^yr_0(t)v_3(t)f(t)dt\,\, \text{ with }
b_j=a_j\text{ or } \breve a_j,
$$
$$
\phi(y)=\int_0^y\left(\pm\sqrt{E-V_1(t)}\pm\sqrt{E-V_2(t)}\right)dt.
$$
Notice that the symbol $\sigma(y)$ and its derivative $D_y\sigma(y)$ are both uniformly bounded with respect to $h$. 
 Lemma~\ref{lem:OI-degm} then implies that this integral is estimated by $\ord(h^{1/(m+1)})$.   
\end{proof}

The following proposition suggests an appropriate choice of the endpoints $s_j$ and $t_j$ for the definition of the $K_j$'s. \begin{proposition}\label{prop:MLbehavior}
Let $K_1$ and $K_2$ be given by \eqref{kj}. If $s_1,$ $s_2$, $t_1$, $t_2$ are negative, then for $j=1,2$, we have
\be
\label{wjflat}
w_j(K_1,K_2,u_j)\equiv
\left\{
\begin{aligned}
& f_j^\inc+\ord (h) \qtext{ near }\edg_j^\inc,\\
&\ord (h) \qtext{ near } \edg_{\hat j}^\inc\cup\breve \edg_1^\out\cup\breve \edg_2^\out,
\end{aligned}
\right.
\ee
\be
w_j(K_1,K_2,\breve u_j)\equiv
\left\{
\begin{aligned}
&\breve{f}_j^\out+\ord (h)\qtext{ near}\breve \edg_j^\out,\\
&\ord (h)\qtext{ near} \breve \edg_{\hat j}^\out\cup\edg_1^\inc\cup\edg_2^\inc.
\end{aligned}
\right.
\ee
 Similarly, if $s_1,$ $s_2$, $t_1$, $t_2$  are positive, then we have
\be
w_j(K_1,K_2, u_j)\equiv
\left\{
\begin{aligned}
& f_1^\out+\ord (h)\qtext{ near} \edg_j^\out,\\
&\ord (h)\qtext{ near} \edg_{\hat j}^\out\cup\breve \edg_1^\inc\cup\breve \edg_2^\inc,
\end{aligned}
\right.
\ee
\be
w_j(K_1,K_2,\breve u_j)\equiv
\left\{
\begin{aligned}
&\breve{f}_j^\inc+\ord (h)\qtext{ near}\breve \edg_j^\inc,\\
&\ord (h)\qtext{ near} \breve \edg_{\hat j}^\inc\cup\edg_1^\out\cup\edg_2^\out.
\end{aligned}
\right.
\ee
\end{proposition}
\begin{proof}
We only show \eqref{wjflat} for $j=1$ supposing $s_2, t_2<0$. 
For this, it is enough to prove a local estimate
\be
\label{wkku}
w_1(K_1,K_2, u_1)=\begin{pmatrix} u_1\\0\end{pmatrix}+\ord(h).
\ee
on an interval $[-\dl,-\dl/2]$ with a small $\delta>0$.
In fact, we see from Proposition~\ref{prop:mlWKB} and \eqref{eq:scalarWKB-allowed} that  the vector-valued function ${}^t(u_1, 0)$ has the same microlocal asymptotic behavior as the right hand side of \eqref{wjflat}.
% By the proof of Proposition~\ref{prop:ConvSer}, we have slightly detailed estimates for $x\in I$,
%\ben
%\left|\til{K}_jf(x)\right|\le C\sup\left\{\left|f(y)\right|;\,\min\{s_j,t_j,x\}\le y\le\max\{s_j,t_j,x\}\right\}
%\een
%and
%\ben
%\left|\til{K}_j\til{K}_{\hat j}\til f(x)\right|\le Ch^{1/(m+1)}\sup\left\{\left|f(y)\right|;\,\min\{s_1,t_1,s_2,t_2,x\}\le y\le\max\{s_1,t_1,s_2,t_2,x\}\right\}.
%\een

To prove \eqref{wkku}, it is enough to show the estimate
\be\label{eq:Ord-h}
\til{K}_2(u_1)(x)=\ord(h) \qtext{uniformly for} x\in [-\dl,-\dl/2].
\ee
%since both $J_1$ and $J_2$ converge and $\til{K}_2J_1=J_2\til{K}_2$, $\til{K}_1J_2=J_1\til{K}_1$. 
If $U_2=r_0(x)$, for simplicity as in the proof of the previous lemma,
$$
\til{K}_2(u_1)(x)=-hK_2U_2 u_1(x)
=\frac i2\left(\breve u_j(x)\int_{s_2}^xr_0(y)u_j(y)u_1(y)dy-u_j(x)\int_{t_2}^xr_0(y)\breve u_j(y)u_1(y)dy\right)
$$
is again an oscillatory integral of type \eqref{Ihint} with a phase function \eqref{+-+-}. The critical point is the crossing point $x=0$, and it is not contained in the integration intervals $[s_2, x], [t_2,x]$ since $s_2, t_2, x$ are all negative. 
Thus  the first assertion of 
Lemma \ref{lem:OI-degm} applies and \eqref{eq:Ord-h} is obtained.
\end{proof}

Define the operators $K_{j,L}$ and $K_{j,R}$ acting on $C(I)$ using \eqref{eq:Def-K-Gen} by
\be\label{eq:DefKjS-A}
K_{j,L}:=K(\breve u_j,u_j,-\dl,-\dl),\quad
K_{j,R}:=K(\breve u_j,u_j,\dl,\dl),
\ee
and solutions $w_j^\out$, $w_j^\inc$ $(j=1,2)$ to the equation $(P-E)w=0$ using \eqref{Def-Sol-Gen} by
\ben
\begin{aligned}
&w_j^\out:=w_j(K_{1,R},K_{2,R},u_j),\quad&&
w_j^\inc:=w_j(K_{1,L},K_{2,L},u_j),\\
&\breve w_j^\out:=w_j(K_{1,L},K_{2,L},\breve u_j),\quad&&
\breve w_j^\inc:=w_j(K_{1,R},K_{2,R},\breve u_j).
\end{aligned}
\een
These solutions are well-defined according to Proposition~\ref{prop:ConvSer}. By Proposition~\ref{prop:MLbehavior},
we know the microlocal asymptotic behavior of $w_j^\inc$ and $\breve w_j^\out$ on the left of $\rho_x(=0)$, and that of $w_j^\out$ and $\breve w_j^\inc$ on the right of $\rho_x$.
In particular, $w_1^\inc$ and $w_2^\inc$ satisfy \eqref{eq:BehaveInc}. We compute the coefficients $t_{j,k}^\out$'s in \eqref{eq:BehaveOut} by representing $w_1^\inc$ and $w_2^\inc$ as linear combinations of the elements of the basis $(\breve w_1^\inc,w_1^\out,\breve w_2^\inc,w_2^\out)$.
\begin{proposition}\label{prop:TransfAlpha}
We have the following linear relationship between the two bases
\be\label{eq:Change-Bases}
(\breve w_1^\out,w_1^\inc,\breve w_2^\out,w_2^\inc)=(\breve w_1^\inc,w_1^\out,\breve w_2^\inc,w_2^\out)A(E,h).
\ee
Here, the $4\times 4$-matrix $A$ is given by $A=(I_4+A_R)^{-1}(I_4+A_L)$ with
\begin{align*}
&A_S=\begin{pmatrix}
\breve \beta_{1,S} \breve u_1	&\breve \beta_{1,S}u_1	&\breve \alpha_{2,S} \breve u_2	&\breve\alpha_{2,S} u_2\\
\beta_{1,S}\breve u_1	&\beta_{1,S}u_1	&\alpha_{2,S}\breve u_2	&\alpha_{2,S}u_2\\
\breve \alpha_{1,S}\breve u_1	&\breve \alpha_{1,S}u_1 	&\breve \beta_{2,S}\breve u_2	&\breve \beta_{2,S}u_2\\
\alpha_{1,S}\breve u_1	&\alpha_{1,S}u_1	&\beta_{2,S}\breve u_2	&\beta_{2,S}u_2
\end{pmatrix},
\quad
S=L,R,
\end{align*}
where $\alpha_{j,S}$, $\breve\alpha_{j,S}$, $\beta_{j,S}$, $\breve\beta_{j,S}$ $(j=1,2,$ $S=L,R)$ are linear functionals defined by
\begin{align*}
&
\alpha_{j,L}f=-{\frac i2}\int_{-\delta}^0	 \breve u_{\hat j}(y)U_{\hat j}J_{j,L}f(y)dy,\quad 
& &\breve\alpha_{j,L}f={\frac i2}\int_{-\delta}^0	 u_{\hat j}(y)U_{\hat j}J_{j,L}f(y)dy,\\
&
\beta_{j,L}f=-{\frac i2}\int_{-\delta}^0 \breve u_j(y)U_j\til{K}_{\hat j,L}J_{j,L}f(y)dy,\quad 
 & &\breve\beta_{j,L}f={\frac i2}\int_{-\delta}^0 u_j(y)U_j\til{K}_{\hat j,L}J_{j,L}f(y)dy,\\
&
\alpha_{j,R}f={\frac i2}\int_0^\delta	 \breve u_{\hat j}(y)U_{\hat j}J_{j,R}f(y)dy,
\quad 
& &\breve\alpha_{j,R}f=-{\frac i2}\int_0^\delta	 \breve u_{\hat j}(y)U_{\hat j}J_{j,R}f(y)dy,\\
&
\beta_{j,R}f={\frac i2}\int_0^\delta \breve u_j(y)U_j\til{K}_{\hat j,R}J_{j,R}f(y)dy,
\quad 
& &\breve\beta_{j,R}f=-{\frac i2}\int_0^\delta u_j(y)U_j\til{K}_{\hat j,R}J_{j,R}f(y)dy,
\end{align*}
where $J_{j,S}=J_j(K_{1,S}, K_{2,S})$.
\end{proposition}

\begin{proof}
Let $\bm{u}_j$ and $\breve{\bm{u}}_j$ be the vectors given by
\ben
\bm{u}_1:=\begin{pmatrix}u_1(0)\\0\\u_1'(0)\\0\end{pmatrix},\quad
\breve{\bm{u}}_1:=\begin{pmatrix}\breve u_1(0)\\0\\\breve u_1'(0)\\0\end{pmatrix},\quad
\bm{u}_2 :=\begin{pmatrix}0\\u_2(0)\\0\\u_2'(0)\end{pmatrix},\quad
\breve{\bm{u}}_2 :=\begin{pmatrix}0\\\breve u_2(0)\\0\\\breve u_2'(0)\end{pmatrix}.
\een
Then by the same arguments as in the proof of Proposition~5.2 of \cite{FMW1}, we have
\begin{align*}
&\begin{pmatrix}
\breve w_1^\out&w_1^\inc&\breve w_2^\out&w_2^\inc\\
(\breve w_1^\out)'&(w_1^\inc)'&(\breve w_2^\out)'&(w_2^\inc)'
\end{pmatrix}\big|_{x=0}
=\begin{pmatrix}\breve{\bm{u}}_1&\bm{u}_1&\breve{\bm{u}}_2&\bm{u}_2\end{pmatrix}(I+A_L),\\
&\begin{pmatrix}
\breve w_1^\inc&w_1^\out&\breve w_2^\inc&w_2^\out\\
(\breve w_1^\inc)'&(w_1^\out)'&(\breve w_2^\inc)'&(w_2^\out)'
\end{pmatrix}\big|_{x=0}
=\begin{pmatrix}\breve{\bm{u}}_1&\bm{u}_1&\breve{\bm{u}}_2&\bm{u}_2\end{pmatrix}(I+A_R).
\end{align*}
This means 
$$
\begin{pmatrix}\breve{\bm{u}}_1&\bm{u}_1&\breve{\bm{u}}_2&\bm{u}_2\end{pmatrix}(I+A_L)=\begin{pmatrix}\breve{\bm{u}}_1&\bm{u}_1&\breve{\bm{u}}_2&\bm{u}_2\end{pmatrix}(I+A_R)A,
$$
and the proposition follows since the matrix $\begin{pmatrix}\breve{\bm{u}}_1&\bm{u}_1&\breve{\bm{u}}_2&\bm{u}_2\end{pmatrix}$ is invertible.
\end{proof}

Writing 
$A=(\bm{a}_1,\bm{a}_2,\bm{a}_3,\bm{a}_4)$, $\bm{a}_j={}^t(a_{j1}, a_{j2}, a_{j3}, a_{j4})$, 
we have obatined
\be
w_1^\inc=(\breve w_1^\inc,w_1^\out,\breve w_2^\inc,w_2^\out)\bm{a}_2,\quad
w_2^\inc=(\breve w_1^\inc,w_1^\out,\breve w_2^\inc,w_2^\out)\bm{a}_4.
\ee
This with Proposition \ref{prop:MLbehavior} shows
\be
T^\out=\begin{pmatrix}a_{22}&a_{24}\\a_{42}&a_{44}\end{pmatrix}+\ord(h).
\ee
The following lemma gives the asymptotic formula for these coefficients, and thus ends the proof of Theorem~\ref{thm:ConnectionAllowed}.

\begin{lemma}\label{lem:ComputeAlpha}
For each $(j,S)\in\{1,2\}\times\{L,R\}$, we have
\be\label{eq:EstimateAlpha}
\alpha_{j,S}u_j,\quad
\alpha_{j,S}\breve u_j,\quad
\breve \alpha_{j,S}u_j,\quad
\breve \alpha_{j,S}\breve u_j=\ord(h^{1/(m+1)}),
\ee
\be\label{eq:EstimateBeta}
\beta_{j,S}u_j,\quad
\beta_{j,S}\breve u_j,\quad
\breve \beta_{j,S}u_j,\quad
\breve \beta_{j,S}\breve u_j =\ord(h^{2/(m+1)}).
\ee
Moreover, with the notations $F_j:=\alpha_{j,L}-\alpha_{j,R}$, $\breve F_j:=\breve\alpha_{j,L}-\breve\alpha_{j,R}$, we have 
\be
A=
\begin{pmatrix}
1&0&\breve F_2\breve u_2&\breve F_2u_2\\
0&1& F_2\breve u_2& F_2u_2\\
\breve F_1\breve u_1&\breve F_1u_1&1&0\\
 F_1\breve u_1& F_1u_1&0&1
\end{pmatrix}
+\ord(h^{2/(m+1)}),
\ee
and
\ben
 F_1u_1=-i\omega h^{1/(m+1)}+\ord(h^{2/(m+2)}),\quad
 F_2u_2=-i\overline{\omega} h^{1/(m+1)}+\ord(h^{2/(m+2)}),
\een
where $\omega = \omega_\rho$ is given by
\be
\label{omega}
\omega_\rho= \mu_m \left(\frac{2(m+1)! }{|v_1-v_2|}\right)^{\frac 1{m+1}}E_0^{-\frac{m}{2(m+1)}}\bm{\Gamma}\left(\frac{m+2}{m+1}\right)
\overline{U(\rho)},
\ee
\be
\label{mu}
\mu_m=\left\{
\begin{aligned}
&e^{\frac{i\pi}{2(m+1)}\ope{sgn}(v_2-v_1)}	& & \qtext{when $m$ is odd,}\\
&\cos\left(\frac{\pi}{2(m+1)}\right)	& & \qtext{when $m$ is even.}
\end{aligned}
\right.
\ee
\end{lemma}

\begin{remark}\label{rem:gencase}
In the general case, where the crossing point is $\rho=(\rho_x,\rho_\xi)\in\mc{V}$, one has to replace $E_0$ by $E_0-V_j(\rho_x)$ and $v_j$ by $V_j'(\rho_x)$. In addition, when $m$ is odd, $\mu_m$ is given by
\ben
\mu_m=e^{\frac{i\pi}{2(m+1)}\ope{sgn}\rho_\xi(v_2-v_1)}.
\een
\end{remark}

\begin{proof}
The estimates \eqref{eq:EstimateAlpha} and \eqref{eq:EstimateBeta} are proved in the same manner as Proposition~\ref{prop:ConvSer}. 

We compute the principal term of $F_1u_1=(\alpha_{1,L}-\alpha_{1,R})u_1$. Since $J_{1,L}$ and $J_{1,R}$ are both identity at the principal level, we have
\ben
F_1u_1
=-{\frac i2}\int_{-\dl}^\dl\breve u_2(y)U_2u_1(y)dy+\ord(h^{2/(m+1)}).
%=\int_{-\dl}^\dl \sigma(y)e^{i\phi(y)/h}dy+\ord(h^{2/(m+1)}),
\een
This is again an oscillatory integral \eqref{Ihint} with the phase function
\ben
\phi(y)=\int_0^y\left(\sqrt{E-V_1(t)}-\sqrt{E-V_2(t)}\right)dt
\een
which behaves like   
\be\label{eq:ApproPhase}
\begin{aligned}
\phi(y)
=\int_0^y\frac{V_2(t)-V_1(t)}{\sqrt{E-V_1(t)}+\sqrt{E-V_2(t)}}dt
=\frac{V_2^{(m)}(0)-V_1^{(m)}(0)}{2\sqrt{E}(m+1)!}\left(y^{m+1}+\ord(y^{m+2})\right)
\end{aligned}
\ee
near the stationary point $0\in I=[-\delta,\delta]$, and the amplitude $\sigma(y,h)$ satisfying
\ben
\sigma(y,h)=-i\frac{(r_0(y)-ir_1(y)\sqrt{E-V_1(y)})}{2(E-V_1(y))^{1/4}(E-V_2(y))^{1/4}}+\ord(h).
\een
Then the asymptotic formula for $F_1u_1$ is computed using Lemma~\ref{lem:OI-degm}.
\end{proof}

\subsection{The space of microlocal solutions near a crossing point (Proof of Theorem~\ref{thm:T-exists})}
\label{section3.4}

Recall from Proposition \ref{prop:MLbehavior} that 
$$
w_j^\flat\equiv
\left\{
\begin{aligned}
& f_j^\inc+\ord (h) \qtext{ near }\edg_j^\inc,\\
&\ord (h) \qtext{ near } \edg_{\hat j}^\inc\cup\breve \edg_1^\out\cup\breve \edg_2^\out,
\end{aligned}
\right.
\quad
\breve w_j^\sharp\equiv
\left\{
\begin{aligned}
&\breve{f}_j^\out+\ord (h)\qtext{ near}\breve \edg_j^\out,\\
&\ord (h)\qtext{ near} \breve \edg_{\hat j}^\out\cup\edg_1^\inc\cup\edg_2^\inc.
\end{aligned}
\right.
$$
We can arrange $w_j^\flat$ so that it is microlocally zero near $\breve \rho$ instead of $\ord(h)$, and similarly arrange $\breve w_j^\sharp$ so that it is microlocally zero near $\rho$.

\begin{lemma}\label{lem:exact-pm}
There exist a basis $(w_1^+,w_2^+,w_1^-,w_2^-)$ of exact solutions to the equation $(P-E)w=0$  such that, for $j=1,2$,
\be\label{eq:exact-pm-in}
w_j^+\equiv \left\{
\begin{aligned}
&w_j^\inc&&\text{near }\rho,\\
&\,\,0&&\text{near }\breve\rho,
\end{aligned}\right.
\quad
w_j^-\equiv \left\{
\begin{aligned}
&\,\,0&&\text{near }\rho,\\
&\breve{w}_j^\out&&\text{near }\breve\rho.
\end{aligned}\right.
\ee
%\be
%\ope{WF}_h(w_j^\pm)\subset\{\pm\xi>0\}\qtext{and}
%\left\{
%\begin{aligned}
%&w_j^+\equiv w_j^\inc\qtext{ near}\rho,\\
%&w_j^-\equiv \breve{w}_j^\out\qtext{ near}\breve\rho.
%\end{aligned}\right.
%\ee
\end{lemma}
%Then our requiring basis is given by
%\be
%w_1:=\frac{w_1^+}{\|w_1^+\|_{L^2(I_\dl)}},\quad
%w_2:=\frac{w_2^+}{\|w_2^+\|_{L^2(I_\dl)}},\quad
%w_3:=\frac{w_1^-}{\|w_1^-\|_{L^2(I_\dl)}},\quad
%w_4:=\frac{w_2^-}{\|w_2^-\|_{L^2(I_\dl)}}.
%\ee
\begin{figure}[t]
\centering
\includegraphics[bb=0 0 349 227, width=8cm]{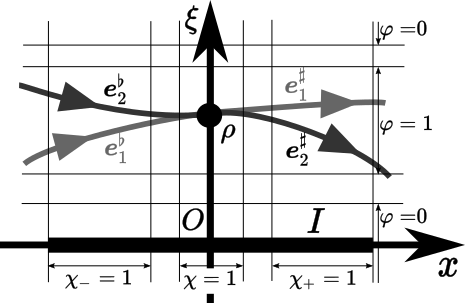}
\caption{Cutoff functions for the construction of a quasi-mode}
\label{Fig:B}
\end{figure}

This lemma immediately implies Proposition~\ref{prop:Basis-pm}. We first prove Proposition \ref{lem:ApproByEx} using this lemma.
\begin{proof}[Proof of Proposition~\ref{lem:ApproByEx}] 
Let $w\in L^2(I;\C^2)$ be a microlocal solution near $\rho$ with $\|w\|_{L^2(I)}\le1$. 
Lemma~\ref{lem:exact-pm} with the microlocal behavior derived by Proposition~\ref{prop:MLbehavior} of $w_j^\inc$ and $\breve w_j^\out$ guarantees that there exist exact solutions $w^\inc$ and $w^\out$ such that 
\ben
w\equiv w^\dir\qtext{microlocally near} \edg_1^\dir\cup \edg_2^\dir,\quad
\ope{WF}_h(w^\dir)\subset\{\xi>0\}
\een
for each $\dir=\out,\inc$. %Fix an interval $I'\subset\joinrel\subset I$. 
We define the function $w_{\text{quasi}}$ by
\be\label{eq:tilU-A}
w_{\text{quasi}}=\chi(x)\pphi(hD_x)w+\chi_-(x)w^\inc+\chi_+(x)w^\out,
\ee
where the functions $\pphi,\chi,\chi_\pm,\in C^\infty(\R;[0,1])$ satisfy $\chi+\chi_-+\chi_+=1$, $\pphi(\xi)=1$ near $\xi=\sqrt{E_0}$, $\chi(x)=1$ near $x=0$, $\chi_\pm(x)=0$ for $\pm x<0$, \ben
\begin{aligned}
&(\ope{supp}\chi\times\ope{supp}\pphi )\cap\ope{WF}_h(w)\subset(\Gamma^I\cap\{\xi>0\})\subset\{\ope{supp}\pphi(\xi)=1\},%,\\
%&(\ope{supp}\chi\times\ope{supp}\pphi' )\cap\ope{WF}_h(w)\cap(\Gamma\cap\{\xi>0\})=\emptyset,
%&(I\times\ope{supp}\pphi')\cap\Gamma=\emptyset.
\end{aligned}
\een
and that the support of $\pphi$, $\chi$ and of $\chi_{I'}$ are compact (see Figure~\ref{Fig:B}). 
More precisely, the operator $\pphi(hD_x)$ is defined by the following sense: 
Fix an interval $I'\subset\joinrel\subset I$, take $\chi_{I'}\in C_0^\infty(I;[0,1])$ such that $\chi_{I'}=1$ on $I'$, and
\ben\pphi(hD_x)f(x):=\frac1{2\pi h}\iint_{\R^2}e^{i(x-y)\xi/h}\pphi(\xi)\chi_{I'}(y)f(y)dyd\xi.\een
In what follows, we denote $I'$ by $I$.

Let us show 
\be\label{eq:ApproByQuasi}
w\equiv w_{\text{quasi}}\qtext{ near}\rho,\quad
\|(P(h)-E)w_{\text{quasi}}\|_{L^2(I)}=\ord(h^\infty)
\ee 
for $w_{\text{quasi}}$ given by \eqref{eq:tilU-A}. By writing 
\ben
(P(h)-E)w_{\text{quasi}}=[P,\chi]\pphi(hD_x)w+\chi[P,\pphi(hD_x)]w+\chi\pphi(hD_x)(P-E)w+[P,\chi_-]w^\inc+[P,\chi_+]w^\out,
\een
we see that the terms $\chi[P,\pphi(hD_x)]w$ and $\chi\pphi(hD_x)(P-E)w$ are $\ord(h^\infty)$ since the essential support of the Weyl symbol of the pseudo-differential operator $\chi[P,\pphi(hD_x)]$ and of $\chi\pphi(hD_x)$ is compact and does not intersect with the semiclassical wavefront set of $w$ and of $(P-E)w$, respectively. 
By the equality $[P,\chi]=-([P,\chi_-]+[P,\chi_+])$, the other terms are sum of the following two terms:
\ben
\begin{aligned}
&[P,\chi_-](w^\inc-\pphi(hD_x)w)=[P,\chi_-](1-\pphi(hD_x))w^\inc+\ord(h^\infty),\\
&[P,\chi_+](w^\out-\pphi(hD_x)w)=[P,\chi_+](1-\pphi(hD_x))w^\out+\ord(h^\infty).
\end{aligned}
\een
Since the functions $w^\inc$ and $w^\out$ are exact solution to $(P-E)w^\dir=0$ and microlocally zero away from $\Gamma^I\cap\{\xi>0\}$, the above two terms are also $\ord(h^\infty)$.

%\ben(1-\pphi(hD_x))w^\dir=\ord(h^\infty).\een Consequently, the above two terms are also $\ord(h^\infty)$.
%The microlocal ellipticity of the operator $(P(h)-E)$ on $\ope{supp}\chi'\times\ope{supp}(1-\pphi)$ with 

%implies that these terms are also of order $h^\infty$. 
%These terms are also of order $h^\infty$ since $w^\dir$ is microlocally $0$ near $\ope{supp}\chi'\times\ope{supp}(1-\pphi)$, and 

By applying Cramer's rule, we have
\be
w_{\text{quasi}}(x)=\frac 1\omega\sum_{k=1}^4 \omega_k(x) w_k,\quad
\omega:=\det
\begin{pmatrix}
w_1&w_2&w_3&w_4\\
w_1'&w_2'&w_3'&w_4'
\end{pmatrix},
\ee
where the function $\omega_k$ is given by the Wronskian of the tuple $(w_1,w_2,w_3,w_4)$ replaced $w_k$ by $w_{\text{quasi}}$ for each $k$, for example,
\ben
\omega_1(x):=\det
\begin{pmatrix}
w_{\text{quasi}}(x)&w_2(x)&w_3(x)&w_4(x)\\
w_{\text{quasi}}'(x)&w_2'(x)&w_3'(x)&w_4'(x)
\end{pmatrix},
\een
Then $\omega_k'=\ord(h^\infty)$ uniformly on $I$, and 
\be
w_{\text{quasi}}(x)=\til{w}+\ord(h^\infty)\qtext{uniformly on}I
\qtext{with}\til{w}=\frac1\omega\sum_{k=1}^4 \omega_k(0)w_k(x)+\ord(h^\infty).
\ee
This with \eqref{eq:ApproByQuasi} proves Proposition \ref{lem:ApproByEx}.
\end{proof}

%Now, in order to prove Lemma~\ref{lem:exact-pm}, we need the following result. 

\begin{proof}[Proof of Lemma~\ref{lem:exact-pm}]
Let $w$ be a bounded exact solution to $(P-E)w=0$ on $I$. 
The asymptotic behavior \eqref{eq:BehaveInc} of $w_1^\inc$ and $w_2^\inc$ implies that, for each $j=1,2$, there exists a linear combination of them which coincides with $f_j^\inc$ microlocally near $\edg_j^\inc$ and with $0$ near $\edg_{\hat j}^\inc$. 
Consequently, there exist constants $c_1$ and $c_2$ such that
\be\label{eq:w0sim0}
w_0:=w-( c_1w_1^\inc+c_2w_2^\inc)\equiv0\qtext{ near} \edg_1^\inc\cup \edg_2^\inc.
\ee 
%\begin{lemma}\label{prop:n-thExact}
We claim that for each $n\in\N$, there exists a basis $(w_{1,n}^-,\,w_{1,n}^+,\,w_{2,n}^-,\,w_{2,n}^+)$ of exact solutions such that 
\ben
w_{j,n}^-\equiv\left\{
\begin{aligned}
&(1+\ord(h^{2/(m+1)}))\breve f_j^\dir&&\qtext{ near}\breve \edg_j^\dir,\\
&\ord(h^{1/(m+1)})\breve f_{\hat j}^\dir&&\qtext{ near}\breve \edg_{\hat j}^\dir,\\
&\ord(h^n)&&\qtext{ near}\rho,
\end{aligned}\right.
\een
and
\ben
w_{j,n}^+\equiv\left\{
\begin{aligned}
&(1+\ord(h^{2/(m+1)}))f_j^\dir&&\qtext{ near} \edg_j^\dir,\\
&\ord(h^{1/(m+1)})f_{\hat j}^\dir&&\qtext{ near} \edg_{\hat j}^\dir,\\
&\ord(h^n)&&\qtext{ near}\breve\rho,
\end{aligned}\right.
\een
for each $(j,\hat j)=(1,2), (2,1)$ and $\dir=\inc,\out$. 
%Moreover, they also satisfy uniformly on $I_\dl$ the estimate
%\ben
%w_{j,n}^\pm=w_{j,n'}^\pm+\ord(h^{n'})\qtext{for any}n\ge n'.
%\een
%\end{lemma}
%Since the space of exact solutions is 4-dimensional, $w_0$ is written as a linear combination of the solutions constructed in Lemma \ref{prop:n-thExact} for any $n\in\N$:
We write $w_0$ as a linear combination of this basis:
\ben
w_0=c_{1,n}w_{1,n}^++c_{2,n}w_{2,n}^++c_{3,n}w_{1,n}^-+c_{4,n}w_{2,n}^-.
\een 
The boundedness of $\|w_0\|_{L^2(I)}$ implies that the coefficients $c_{k,n}$ $(k=1,2,3,4)$ are bounded as $h\to0^+$, and in particular, we have
\be\label{eq:unique-exactA}
\begin{aligned}
w_0-(c_{1,n}w_{1,n}^++c_{2,n}w_{2,n}^+)
=c_{3,n}w_{1,n}^-+c_{4,n}w_{2,n}^-
\equiv\ord(h^n)\text{ near }\rho.
\end{aligned}
\ee
%Here we have used the microlocal behavior of $w_{1,n}^-$ and $w_{2,n}^-$ given in Lemma \ref{prop:n-thExact}. 
The above formula \eqref{eq:unique-exactA} with \eqref{eq:w0sim0} implies that $c_{1,n}$ and $c_{2,n}$ are $\ord(h^n)$ as $h\to0^+$, and 
\be
w_0%\equiv c_{1,n}w_{1,n}^++c_{2,n}w_{2,n}^++\ord(h^n)
\equiv\ord(h^n)
\qtext{ near} \rho. 
%\edg_1^\out\cup\edg_2^\out.
\ee
Since $n\in\N$ is arbitrarily chosen, we conclude here that 
\be
w\equiv c_1w_1^\inc+c_2w_2^\inc\qtext{ near} \rho.%\edg_1^\out\cup \edg_2^\out.
\ee

In the above argument, we also proved for any exact solution $w$ the existence of an exact solution $w_0$ such that
\be
%\ope{WF}_h(w_0)\subset\{\xi<0\},\quad
w_0\equiv\left\{
\begin{aligned}
&0&&\text{near }\rho,\\
&w+\ord(h)&&\text{near }\breve \rho.
\end{aligned}\right.
\ee
%By applying this to $\breve w_1^\out$ and to $\breve w_2^\out$, w
We see that the required solutions $w_1^-$ and $w_2^-$ are given by linear combinations of $w_0$ corresponding to $\breve w_1^\out$ and to $\breve w_2^\out$. 
The solutions $w_1^+,w_2^+$ are given in the symmetric way. 
%A proof for $w_1^+,w_2^+$ is given in the same way. 

We then prove the existence of such a basis $(w_{1,n}^-,w_{1,n}^+,w_{2,n}^-,w_{2,n}^+)$. We here only discuss on $w_{1,n}^+$. Others are also shown in a symmetric way. %the same way. 
It suffices to construct an approximate solution $g_k$ such that
\be
(P(h)-E)g_k=\begin{pmatrix}hU_1r_k\\0\end{pmatrix}
\ee
with a function $r_k$ which is of order $h^{(2k-1)/(m+1)}$ uniformly on $I$. 
In particular, $r_k=\ord(h^n)$ for a large enough $k\ge((m+1)n+1)/2$. 
Then the required solution $w_{1,n}^+$ is given by
\be
w_{1,n}^+=g_k+\begin{pmatrix}J_{1,L}r_k\\\til{K}_{2,L}J_{1,L}r_k\end{pmatrix}.
\ee

Such an approximate solution $g_k$ is constructed inductively. 
The facts
\ben
(P_2-E)\til{K}_{2,L}u_1=-hU_2u_1,\quad
\ope{WF}_h(-hU_2u_1)\subset\ope{WF}_h(u_1)\subset\Gamma_1\cap\{\xi>0\}
\een
imply that $\ope{WF}_h(\til{K}_{2,L}u_1)$ is contained in $(\Gamma_1\cap\{\xi>0\})\cup\Gamma_2$.  
The usual propagation of singularities theorem with the fact $\ope{WF}_h(-hU_2u_1)\cap(\Gamma_2\cap\{\xi<0\})=\emptyset$ shows that there exists a constant $c_{2,1}\in\C$ such that 
\ben
\til{K}_{2,L}u_1\equiv c_{2,1} \breve u_2\qtext{microlocally near} \Gamma_2\cap\{\xi<0\}.
\een 
According to the estimate \eqref{eq:Ord-h}, this constant admits $c_{2,1}=\ord(h)$ as $h\to0^+$.  
By putting $g_{1,1}:=u_1$ and $g_{2,1}:=\til{K}_{2,L}u_1-c_{2,1}\breve u_2$, we have
\be
\ope{WF}_h(g_{1,1})
\subset
\Gamma\cap\{\xi>0\},\quad
\ope{WF}_h(g_{2,1})\subset\Gamma\cap\{\xi>0\}.
\ee 
We then apply the above argument for $g_{1,1}=u_1$ to $g_{2,1}$. 
Then there exists $c_{1,2}=\ord(h^2)$ such that one has
\be
\ope{WF}_h(g_{1,2})\subset\Gamma\cap\{\xi>0\}\qtext{for}
g_{1,2}:=g_{1,1}+\til{K}_{1,L}g_{2,1}-c_{1,2} \breve u_1.
\ee 
By repeating this process, we obtain such functions $g_{1,k}$, $g_{2,k}$ and  constants $c_{j,k}=\ord(h^{2k+j-3})$ $(j=1,2)$ for each $k\in\N$ that we have 
\be
\begin{aligned}
&\ope{WF}_h(g_{1,k})\subset\Gamma\cap\{\xi>0\}
\qtext{for}
g_{1,k}:=g_{1,k-1}+\til{K}_{1,L}g_{2,k-1}-c_{1,k} \breve u_1,\\
&\ope{WF}_h(g_{2,k})\subset\Gamma\cap\{\xi>0\}
\qtext{for}
g_{2,k}:=g_{2,k-1}+\til{K}_{2,L}g_{1,k}-c_{2,k} \breve u_2.
\end{aligned}
\ee
Then these functions satisfy
\ben
(P(h)-E)g_k
=\begin{pmatrix}h U_1r_k\\0\end{pmatrix}\qtext{with}
g_k=\begin{pmatrix}g_{1,k}\\g_{2,k}\end{pmatrix},\quad
r_k=\til{K}_{2,L}g_{1,k}-c_{2,k}\breve u_2.
\een
The estimates on the resolvent operators given by Proposition~\ref{prop:ConvSer} implies $r_k=\ord(h^{(2k-1)/(m+1)})$. 
\end{proof}

\section{Resonance asymptotics}\label{sec:prfMain}
This section is devoted to the proof of Theorem~\ref{MAINTH}. 

\subsection{Probability amplitude}
We call a complex number $E\in\cR$ a \textit{pseudo-resonance} of $P$ if there exists an outgoing microlocal solution to the system $(P-E)w=0$ near an enough large bounded subset of $\Gamma$. We make this definition precise by introducing the notions of \textit{probability amplitude} and \textit{monodromy matrix}. We will see that the pseudo-resonances approximate the resonances in the semiclassical limit.

On the characteristic set $\Gamma$, we define what we call \textit{probability amplitude}. The probability amplitude is a complex number associated with each \textit{generalized classical trajectory} (see Definition \ref{GCT}). 
%In the following, we identify a generalized classical trajectory $\gamma:[0,T]\to \Gamma$ with its image $\gamma([0,T])\subset \Gamma$. 
We recall that $S_{\gamma}=\int_{\gamma}\xi dx$ is the action along the trajectory $\gamma$.

Recall first from the Maslov theory that the microlocal WKB solutions change their asymptotic behavior when continued beyond a turning point  (see \cite{FMW3} in our system case). 
More precisely, 
let $\gamma_{\rm{cl}}$ be a classical trajectory for $p_j$ starting from $\rho_-=(x_-,\xi_-)\in\Gamma_j\setminus(\mc{V}\cup\{\xi=0\})$ and ending at $\rho_+=(x_+,\xi_+)\in\Gamma_j\setminus(\mc{V}\cup\{\xi=0\})$ without passing through any crossing point: $\gamma_{\rm{cl}} \cap \mathcal{V}= \emptyset$. Then there exists 
$\nu_{\gamma_{\rm{cl}}}=\nu_{\gamma_{\rm{cl}}}(h)$ which coincides with the number of turning points in $\gamma_{\rm{cl}}$ modulo $\ord(h)$ 
such that the two microlocal solutions $f_{\rho_-}(x,h;x_-)$ and 
$f_{\rho_+}(x,h;x_+)$ to the system near $\rho_-$ and $\rho_+$, respectively (see Proposition~\ref{prop:mlWKB}), are related with the formula:
\be\label{eq:Maslov-nu1}
f_{\rho_-}(x,h;x_-)=\exp\left(\frac ihS_{\gamma_{\rm{cl}}}%\int_\gamma\xi dx
-\frac{i\pi }2 \nu_{\gamma_{\rm{cl}}}\right)f_{\rho_+}(x,h;x_+).
\ee
%The associated probability amplitude $\mathscr{P}(\gamma_{\rm{cl}})$ is simply defined by 
%$$
%\mathscr{P}(\gamma_{\rm{cl}}) := \exp\left(\frac ihS_{\gamma_{\rm{cl}}}%\int_\gamma\xi dx
%-\frac{i\pi }2 \nu_{\gamma_{\rm{cl}}}\right).
%$$
%For a generalized classical trajectory $\gamma$, one has to take the transfer matrices at crossing points into account in order to define the associated probability amplitude. 

Let $\gamma: [0,T]\to \Gamma$ be a generalized classical trajectory whose extremities are supposed not to be a crossing point nor a turning point. Thanks to Assumption \ref{H3} which guarantees that a turning point cannot be a crossing point,  one can take a finite partition of the time interval 
\be\label{eq:TimePartition}
[0,T]= \bigsqcup_{l \in L} I_{l} \sqcup \bigsqcup_{n \in N} I_{n}
\ee
such that, for each $l\in L$, $\gamma_l:=\gamma(I_l)$  passes  through a single turning point and no crossing point, and for $n\in N$, $\gamma_n:=\gamma(I_n)$ passes through a single crossing point $\rho_n\in\mc{V}$ and no turning point.
%More precisely  one writes
%$$
%\gamma= \bigsqcup_{l \in L} \gamma_{l} \sqcup \bigsqcup_{n \in N} \gamma_{n}
%$$ 
%with $\{\gamma_l =  \gamma_{\rm cl}^l , l \in L\}$ and $\{\gamma_n, n \in N \}$ satisfying
%$$
%\#(\gamma_l \cap \{\xi=0\})=1 \quad ; \quad \gamma_l \cap \mathcal{V} = \emptyset, \;\; \forall l \in L,
%$$
%$$
% \gamma_n \cap \{\xi=0\} = \emptyset \quad ; \quad  \gamma_n \cap \mathcal{V} = \{\rho_n\}, \quad \forall n \in N.
%$$
%where each $ \gamma_{\rm cl}^l$ with $l\in L$ pass through a single turning point and no crossing point while each $\gamma_{\rm gen}^n$ with $n\in N$ pass through a single crossing point and no turning point, and that their extremities are neither a crossing point nor a turning point. 
Let $\nu_{\gamma_l} = 1+ \mathcal{O}(h)$, $l\in L$, be defined by \eqref{eq:Maslov-nu1}. For each $n\in N$, let $T_{\rho_n}$ be the transfer matrix at $\rho_n$ defined with the microlocal WKB solutions at the extremities of $\gamma_n$, given by Theorem \ref{thm:T-exists}. We denote $\tau_{\gamma_n}$ the entry of $T_{\rho_n}$ corresponding to the trajectory $\gamma_n$, that is, if the incoming part of $\gamma_n$ belongs to $\Gamma_j$ and the outgoing part belongs to $\Gamma_k$, then $\tau_{\gamma_n}=(T_{\rho_n})_{k,j}$ is the $(k,j)$-entry of $T_{\rho_n}$.
The following definition of the probability amplitude does not depend of the choice of the partition of the time interval \eqref{eq:TimePartition}. 
\begin{definition}\label{def:PA}
For a generalized classical trajectory $\gamma$ whose extremities are neither a turning point nor a crossing point, we define the \textit{probability amplitude} $\mathscr{P}(\gamma)$ by
\be
\mathscr{P}(\gamma):=\exp\left(\frac ih S_\gamma-\frac{i\pi}2 \sum_{l\in L}\nu_{\gamma_l}\right)
\prod_{n\in N}\tau_{\gamma_n}.
\ee
\end{definition}
\begin{remark}
The probability amplitude maps the addition of generalized classical trajectories to the multiplication. More precisely, if $\gamma_1:[0,T_1]\to\Gamma$ and $\gamma_2:[0,T_2]\to\Gamma$ are two consecutive generalized classical trajectories, that is, the terminal point $\gamma_1(T_1)$ coincides with the initial point $\gamma_2(0)$, then one has
\be\label{eq:multiPA}
\mathscr{P}(\gamma_1\cup\gamma_2)=\mathscr{P}(\gamma_1)\mathscr{P}(\gamma_2).
\ee
Here $\gamma_1\cup\gamma_2:[0,T_1+T_2]\to\Gamma$ stands for the usual composition:
\ben
(\gamma_1\cup\gamma_2)(t)=\left\{
\begin{aligned}
&\gamma_1(t)&&0\le t<T_1,\\
&\gamma_2(t-T_1)&&T_1\le t\le T_1+T_2.
\end{aligned}\right.
\een
\end{remark}
Using the terminology of the graph theory, we call \textit{vertex} a crossing point and \textit{edge} a classical trajectory joining two adjacent vertices (i.e. an edge does not contain any vertex except at the extremities),  and write $\mc E$ the set of all edges (we already use the notation $\mc V$ for the set of vertices). For an edge $\edg$, we denote $\edg^-$ its starting point and $\edg^+$ its endpoint. 

\subsection{Monodromy matrix}
We take a base point $\rho_\edg$ on each edge $\edg \in \mc E$, and define a matrix $M(E,h)$ the size of which is the total number of edges:
\be
M(E,h)=(m_{\edg,\edg'}(E,h))_{\edg,\edg'\in\mc E}.
\ee
The $(\edg,\edg')$-entry is 0 if $\edg^-\neq (\edg')^+$ and otherwise it is given by the probability amplitude $\mathscr{P}({\gamma})$ of the generalized classical trajectory $\gamma$ from $\rho_{\edg'}$ to $\rho_\edg$ passing through exactly one crossing point $\edg^-= (\edg')^+$.

\begin{definition}
We call zeros in $\cR$ of $\det(I-M(E,h))$ pseudo-resonances of $P$, and denote the set of pseudo-resonances by $\ope{Res}_0(P)$. 
\end{definition}
\begin{remark}
The pseudo-resonances are independent of the choice of the base points $\rho_\edg$.
\end{remark}
\begin{remark}
If $E\in\cR$ and if $w$ is an outgoing microlocal solution to the equation $(P-E)w=0$ near some neighborhood of $\mc{E}$, then $E$ is a pseudo-resonance, and there exists an eigenvector $\alpha=(\alpha_\edg)_{\edg\in\mc{E}}$ of the eigenvalue 1 of $M(E,h)$ such that
\be\label{eq:EVofM}
w\equiv \alpha_\edg f_{\rho_\edg}\qtext{ near}\edg
\ee
for any $\edg\in\mc{E}$.  Here $f_{\rho_\edg}$ is the WKB microlocal solution to the system $(P-E)w = 0$ near $\rho_\edg$ given by Proposition \ref{prop:mlWKB}.
\end{remark}
We will see in the next subsection that the pseudo-resonances are a good  approximation of the resonances in $\cR$.
Here we compute the asymptotic behavior of the determinant
$\det(I-M(E,h))$ and observe that the condition  of the pseudo-resonances coincide with the Bohr-Sommerfeld quantization condition at principal level.
\begin{proposition}\label{prop:asympt-qc}
As $h\to 0^+$, we have
\be
\begin{aligned}
&
\det(I-M(E,h))=1+e^{i\mc{A}(E)/h}+\ord(h^{\frac 2{m_0+1}}),\\
&
h\p_E\det(I-M(E,h))=i\cA'(E)e^{i\cA(E)/h}+\ord(h^{\frac 2{m_0+1}}),
\end{aligned}
\ee
uniformly for $E\in\cR$.
\end{proposition}

As a consequence of the above asymptotics one immediately gets the following result:
\begin{corollary}\label{cor:pseudoresonance-nearE:bis}
The pseudo-resonances are simple zeros of $\det(I-M(E,h))$ and there exists a bijection $\til{z}_h:\mathfrak{B}_h\to\ope{Res}_0(P)$ such that as $h\to 0^+$ one has, \be\label{eq:pseudo-resonance-near-E}
\til{z}_h(E)=E+\ord(h^{\frac{m_0+3}{m_0+1}}),
\ee
uniformly for any $E\in\mathfrak{B}_h$. 
\end{corollary}

As in the standard graph theory, we call a consecutive sequence of edges a \textit{path}, and a path whose initial and terminal point coincide a \textit{cycle} (we identify two cycles by ignoring their initial and terminal point). A cycle is said to be \textit{primitive} if it does not contain any sub-cycles. Let $\mc{C}$ be the set of primitive cycles.
\begin{proof}[Proof of Proposition~\ref{prop:asympt-qc}]
By definition, each non-zero entry of $M(E,h)$ stands for the probability amplitude. According to \eqref{eq:multiPA}, the multiplication of probability amplitudes corresponds to the composition of consecutive generalized classical trajectories. 
Consequently, the determinant $\det(I-M)$ can be written in terms of the probability amplitudes associated with the cycles as follows: 
\be\label{eq:Det-PA}
\det(I-M)=1+\sum_{K}\prod_{{\gamma}\in K}(-\mathscr{P}({\gamma})),
\ee
where the summation is taken over the nonempty subsets $K$ of $\mc{C}$ such that ${\gamma}\cap{\gamma}'=\emptyset$ holds for any ${\gamma},$ ${\gamma}'\in K$ with ${\gamma}\neq{\gamma}'$. 
According to Theorem~\ref{thm:ConnectionAllowed}, for each crossing point $\rho\in\mc{V}$, the $(j,k)$-entry of $T_\rho$ is $1+\ord(h^{2/(m_\rho+1)})$ if $j=k$ and is $\ord(h^{1/(m_\rho+1)})$ if $j\neq k$. Then by the definition of the probability amplitude, 
\be\label{eq:ChangeCycle}
\mathscr{P}({\gamma})=\ord(h^{2/(m_0+1)})
\ee
holds for any primitive directed cycle ${\gamma}$ containing some edges of $\mc{E}_2$. We also have
\be\label{eq:PrinCycle}
\mathscr{P}(\Gamma_1)=(-1+\ord(h^{2/(m_0+1)}))e^{i\cA(E)/h}.
\ee
The representation \eqref{eq:Det-PA} of $\det(I-M)$ with Formulae~\eqref{eq:ChangeCycle} and \eqref{eq:PrinCycle} imply the asymptotic formula for $\det(I-M)$. For the derivative, we apply the fact that multiplying by $h$ the derivative with respect to $E$ of each entry of the transfer matrix $T_\rho$ keeps the same order as itself.
\end{proof}

\begin{example}
Let us consider the model in Subsection~\ref{subSec:App}. We label the three edges as in Figure~\ref{Fig:Connection}, and for each $j\in \{1,2,3\}$ we take the base point $\rho_{\edg_j}=\edg_j^-$. Then the monodromy matrix $M$ is given by $M=T\Phi$, where $T$ and $\Phi$ are the $3\times 3$ matrices given by
\be
\Phi=\exp\left(\frac ih \ope{diag}\left(\cA_{1,L}-\frac{h\pi}{2},\cA_{1,R}-\frac{h\pi}{2},\cA_{2,R}-\frac{h\pi}{2}\right)\right),
\ee
with
\ben
\cA_{1,L}:= 2\int_{a(E)}^0\sqrt{E-V_1(x)}\,dx,\quad
\cA_{1,R}:= 2\int_0^{b(E)}\sqrt{E-V_1(x)}\,dx,
\een
\ben
\cA_{2,R}:=2\int_0^{c(E)}\sqrt{E-V_2(x)}\,dx,
\een
and 
\be
T=
\begin{pmatrix}
0&(T_{\rho_-})_{1,1}&(T_{\rho_-})_{1,2}\\
(T_{\rho_+})_{1,1}&0&0\\
(T_{\rho_+})_{2,1}&0&0
\end{pmatrix}.
\ee
We have $\cA_{1,L}+\cA_{1,R}=\cA$, and $\cA_{1,L}+\cA_{2,R}=S_\gamma$ where $S_\gamma$ is the action along the directed cycle  $\gamma$ defined for this example by \eqref{eq:SforEx}, and obtain
\begin{align*}
\det(I-M)=1+e^{i\cA/h}(T_{\rho_+})_{1,1}(T_{\rho_-})_{1,1}
+e^{i%(\cA_{1,L}+\cA_{2,R})/h}
S_\gamma/h}(T_{\rho_+})_{2,1}(T_{\rho_-})_{1,2}.
\end{align*}
\end{example}
%===============	Figure  Connection	==================
\begin{figure}
\centering
\includegraphics[bb=0 0 1372 443, width=15.5cm]{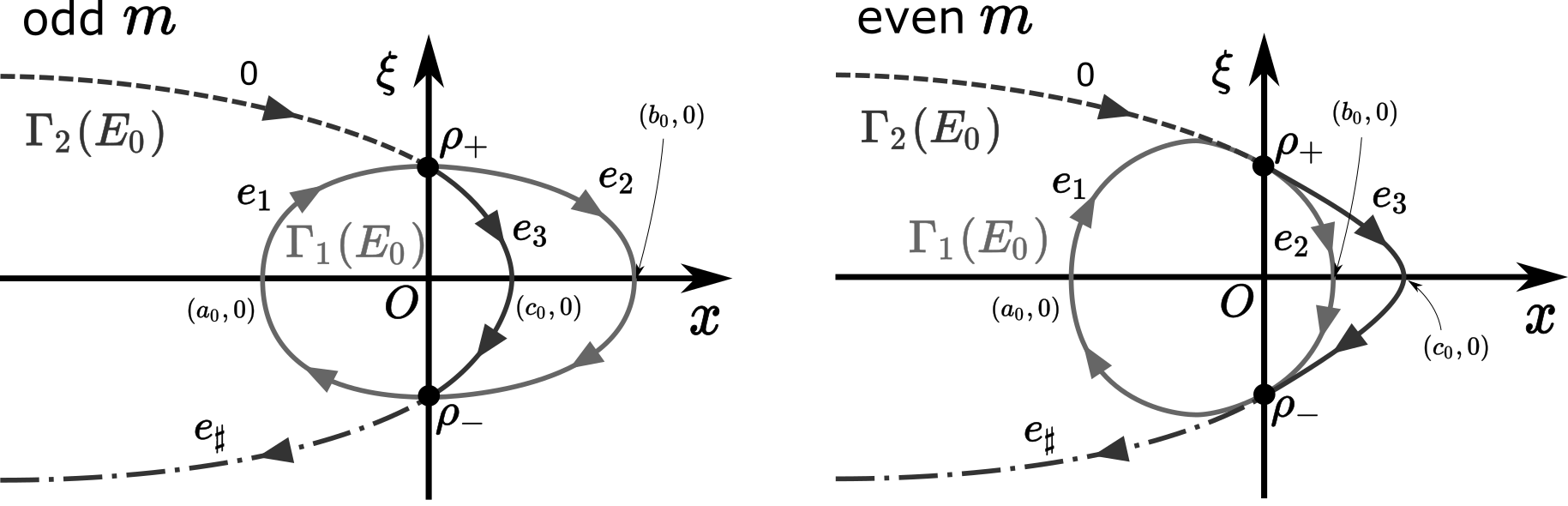}
\caption{Labeling the edges in typical examples}
\label{Fig:Connection}
\end{figure}

\subsection{Bijection between resonances and pseudo-resonances}

In  this subsection, we prove a one to one correspondence between the set of resonances and the  set of pseudo-resonances in the domain $\mathcal R$ and estimate the distance between them.

\begin{proposition}\label{prop:non-ex-res}
For any $h$-independent $C>0$ and for $h$ small enough, there is no resonance outside a $Ch^2$-neighborhood of pseudo-resonances.
\end{proposition}

\begin{remark}
In fact, for any real number $n$, there is no resonance outside a $Ch^n$-neighborhood of pseudo-resonances for $h$ small enough.
\end{remark}
\begin{proof}
We proceed by a contradiction argument. 
Suppose that the proposition is not true. Then there would exist a subset $\mathfrak{h}\subset(0,1]$ and maps $E$ and $w$ defined on $\mathfrak{h}$ with values in $\{z\in\cR;\,\ope{dist}(z,\ope{Res}_0(P(h)))>Ch^2\}$ and $L^2(\R;\C^2)$ respectively, such that $0\in\bar{\mathfrak{h}}$, $\|w(h)\|_{L^2}=1$ and 
\be\label{eq:EFofP}
(P_\theta(h)-E(h))w(h)=0\qtext{for}h\in\mathfrak{h}.
\ee
Then the equality \eqref{eq:EVofM} for $w$ holds with a vector $\alpha=(\alpha_\edg)_{\edg\in\mc{E}}$ such that $(I-M(E,h))\alpha=\ord(h^\infty)$. According to Proposition~\ref{prop:asympt-qc}, the inverse $(I-M(E,h))^{-1}$ exists and is polynomially bounded with respect to $h$ for any $E\in\cR$ with $\ope{dist}(E,\ope{Res}_0(P(h)))>Ch^2$, and consequently $w\equiv0$ near $\Gamma(E_0)$. 
This with Formula \eqref{eq:EFofP} implies that $\|w\|_{L^2}=\ord(h^\infty)$ which contradicts with $\|w\|_{L^2}=1$. 
\end{proof}

Now, using Proposition~\ref{prop:non-ex-res} and the Kato-Rellich theorem we prove the unique existence of a resonance in $\ord(h^2)$-neighborhood of each pseudo-resonance: 
\begin{proposition}\label{prop:existsnce-one-to-one:bis}
For each $E\in\mathfrak{B}_h$, there uniquely exists a resonance $z_h(E)$ verifying
\be\label{eq:resonance-near-PR}
\left|z_h(E)-\til{z}_h(E)\right|\le h^2,
\ee
for small $h>0$.
\end{proposition}
To prove this proposition, we introduce the family  $P_\theta(s,h)$, $s\in\C,\,\left|s\right|<2$ of operators  defined by
\be
\begin{aligned}
&P_\theta(s,h):=P_\theta^{\ope{diag}}(h)+shQ_\theta,\quad
P_\theta^{\ope{diag}}(h):=\mathcal{U}_\theta P^{\ope{diag}}(h) \mathcal{U}_\theta^{-1},\quad
Q_\theta:=\mathcal{U}_\theta Q \mathcal{U}_\theta^{-1},\\
&P^{\ope{diag}}(h):=
\begin{pmatrix}
P_1(h)&0\\0&P_2(h)
\end{pmatrix},\quad
Q:=
\begin{pmatrix}
0&U_1\\U_2&0
\end{pmatrix},\quad
\mathcal{U}_\theta w:=\left|\zeta_\theta'(x)\right|^{1/2}w\circ\zeta_\theta,
\end{aligned}
\ee
where $\zeta_\theta$ is defined in Section~\ref{sec:results}. 
This is first defined for analytic functions in $\Sigma$, and then extended to an operator on $L^2(\R;\C^2)$ by using the analyticity of the coefficients (Assumption~\ref{H1}). 
The operator $P_\theta$ is non-self-adjoint in $L^2(\mathbb R;\mathbb C^2)$, and its spectrum is discrete in $\cR$ due to the assumption $E_0 \neq V_j^\pm$, $j=1,2$ (see \cite[Appendix B]{Hi1}), and is independent of $\theta$. The resonances in $\cR$ of $P$ are equivalently defined as the eigenvalues of $P_\theta$.

For each fixed small $h>0$, this is an analytic family in the sense of Kato, that is, $P_\theta(s)=P_\theta(s,h)$ satisfies the following conditions:
\begin{enumerate}
\item $P_\theta(s)$ is closed and the resolvent set of $P_\theta(s)$ is not empty for each $\left|s\right|<2$.
\item For any $\left|s_0\right|<2$, there exists an element $\lambda_0$ of the resolvent set of $P_\theta(s_0)$ such that $(P_\theta(s)-\lambda_0)^{-1}$ is an analytic operator-valued function of $s$ near $s_0$.
\end{enumerate}
In fact, for fixed small $h$, we can take $C_0>0$ large enough such that 
\be
\left\|\left(P_\theta^{\ope{diag}}-(E_0-iC_0h)\right)^{-1}\right\|_{L^2\to H^2}
\le (h\|Q\|_{H^1\to L^2})^{-1},
\ee
since we know that the spectrum near $E_0$ of $P_\theta(0)$ consists only of the eigenvalues of $P_1$ on the real line, 
the essential spectrum and the resonances of $P_2$ are far at least of order $h\log(1/h)$.
Then the resolvent of $P_\theta(s)$ is given by the Neumann series which depends analytically on $s$:
\ben
(P_\theta(s)-\lambda_0)^{-1}
=\left(I+sh(P_\theta^{\ope{diag}}-\lambda_0)^{-1}Q_\theta\right)^{-1}
(P_\theta^{\ope{diag}}-\lambda_0)^{-1},\quad(\lambda_0=E_0-iC_0h).
\een

Our previous arguments for $P(h)$ and $P_\theta(h)$ remain true for
\ben
P(s,h):=P^{\ope{diag}}+shQ=
\begin{pmatrix}
P_1&shU_1\\shU_2&P_2
\end{pmatrix}\quad(s\in[0,1])
\een 
and $P_\theta(s,h)$ $(s\in[0,1])$ (we may replace $U_1,$ $r_k$ by $sU_1$, $sr_k$ $(k=0,1)$). %, and the dependence on $s$ of $h_0$ is always continuous. 
Note that our considering operator and the reference diagonal operator are $P(h)=P(1,h)$ and $P^{\ope{diag}}=P(0)$, respectively. 
Then we conclude by applying the Kato-Rellich theorem %(Proposition~\ref{prop:KRThm}) 
that for each $E=E(h)\in\mathfrak{B}_h$, there uniquely exists an analytic function $z_h(s,E(h))$ of $s$ such that $z_h(0,E(h))=E(h)+\ord(h^2)$ is an eigenvalue of $P_1$ and $z_h(s,E(h))$ is an eigenvalue of $P_\theta(s,h)$, that is, a resonance of $P(s,h)$.

\begin{proof}[Proof of Proposition~\ref{prop:existsnce-one-to-one:bis}] 
Proposition~\ref{prop:asympt-qc} shows that each pseudo-resonance is away of order $h$ from the others. 
Then the complement in $\cR$ of the resonance free region (shown by Proposition~\ref{prop:non-ex-res} with $C=1$)
\ben
\cR\setminus\{\ope{dist}(z,\ope{Res}_0(P(h)))< h^2\}
\een
is written in the disjoint sum
\ben
\{\ope{dist}(z,\ope{Res}_0(P(h)))< h^2\}
\bigsqcup_{E\in\mathfrak{B}_h}\{\left|z-\til{z}_h(s,E)\right|< h^2\}.
\een
In particular for $s=0$, each eigenvalue $z_h(0,E(h))=E(h)+\ord(h^2)$ of $P_1$ has to satisfy $\left|z_h(0,E(h))-\til{z}_h(0,E)\right|< h^2$. 
We also know that there is no other eigenvalue of the diagonal operator $P_\theta(0,h)$ in $\cR$. Then by the analyticity of $z_h(s,E)$, it is the unique resonance with $\left|z_h(s,E)-\til{z}_h(s,E)\right|< h^2$ for any $s\in[0,1]$. 
\end{proof}

\subsection{Precise width of resonances}\label{sec:precise}
Proposition~\ref{prop:existsnce-one-to-one:bis} with Corollary~\ref{cor:pseudoresonance-nearE:bis} immediately implies that each resonance $z_h(E)$ corresponding to $E\in\mathfrak{B}_h$ admits the asymptotic formula
\be
z_h(E)=E+\ord(h^{\frac{m_0+3}{m_0+1}}).
\ee

The imaginary part of the subprincipal term of $\ord(h^{(m_0+3)/(m_0+1)})$ of this formula would give the width of the resonances.
To obtain this subprincipal  term, we need to compute the subpricipal term of $\ord(h^{2/(m_0+1)})$ of the diagonal entries of the transfer matrices.

However we here employ an alternative method  used in \cite{FMW3}.
It consists in the expression of the width  of resonances in terms of the  resonant state $w$ near infinity. Due to the outgoing property of the resonant state, it is enough to compute the microlocal asymptotics of the resonant state on the outgoing tails of $\Gamma$. We restate this fact in the following lemma.
Let $\rho_1=(x_1,\xi_1)$ [resp. $\rho_2=(x_2,\xi_2)$] be a point on the outgoing tail tending to $-\infty$ [resp. $+\infty$], if it exists, such that $x_1<a_0<b_0<x_2$. 

\begin{lemma}
Let  $f_k$ be the microlocal WKB solution defined in Proposition \ref{prop:mlWKB} near $\rho_k$, whose phase is based on $\rho_k$,  and $\alpha_k$ defined by
\be
\label{wkfk}
w\equiv \alpha_kf_k\text{ near }\rho_k, \quad k=1,2.
\ee
Then we have
\be
\im z=-\frac {h(1+\ord(h))}{\Vert w\Vert^2_{L^2(x_1,x_2)}}(|\alpha_1|^2+|\alpha_2|^2),
\ee
with the convention that $\alpha_k=0$ when $\Gamma_2\cap\{x=x_k\}=\emptyset$.
\end{lemma}
\begin{proof}
We have, by integration by parts from $\im\langle (P(h)-z)w,w\rangle_{L^2(x_1,x_2)}=0$, 
\be\label{eq:Green}
(\im z)\|w\|^2_{L^2(x_1,x_2)}
=h^2\im\left[-v_1'\bar{v}_1-v_2'\bar{v}_2+r_1v_2\bar{v}_1\right]^{x_2}_{x_1},\quad
w=\begin{pmatrix}v_1\\v_2\end{pmatrix},
\ee
where $[g]_{x_1}^{x_2}=g(x_2)-g(x_1)$. 

For each $k=1,2$, $w=\ord(h^\infty)$ near $x_k$ if $\Gamma_2\cap\{x=x_k\}=\emptyset$. 
Otherwise, the line $\{x=x_k\}$ intersects with $\Gamma_2$ at two points (since $\Gamma_2$ is symmetric with respect to $\{\xi=0\}$), where one is on an outgoing trajectory corresponding to $p_2$ while the other is on an incoming trajectory. 
We denote by $\rho_k=(x_k,\xi_k)$, $k=1,2$, the one on the outgoing trajectory.

Remark that in the subset $\omega_k\times \R_\xi$ of the phase space, where $\omega$ is a small neighborhood of $x_k$, $w$ has its semiclassical wave front set only on the outgoing tail.
This fact together with
the microlocal expression of $w$ on the outgoing tail implies that $w$ admits the local expression in $\omega_k$
\be
w=\alpha_{k}\chi^w f_{k}+\ord(h^\infty)
\ee
where $\chi\in C_0^\infty(T^*\R;\R)$ is identically 1 and supported only near $\rho_k$. 
Since $v_1$ is $\ord (h)$ on $\Gamma_2$, we have 
\ben
-v_1'\bar{v}_1+r_1v_2\bar{v}_1\Big|_{x=x_k}=\ord(h),\quad v_2'\bar{v}_2\Big|_{x=x_k}
=(-1)^k
i h^{-1}\left|\alpha_{k}\right|^2(1+\ord(h)).
\een
Thus we obtain the lemma.
\end{proof}

Thanks to this lemma, it is enough to compute the asymptotics of
$\|w\|_{L^2(x_1,x_2)}$ and $\alpha_k$. 

Let $\edg_0$ be the edge on $\Gamma_1$  such that its endpoint $\edg_0^+$ coincides with the starting point of the outgoing tail containing $\rho_1$ on it. We normalize $w$ in such a way that
\begin{equation}
\label{micronormal}
w\equiv f_{\rho_0}\,\,\,\text{ near }\rho_0,
\end{equation}
where $\rho_0$ is the base  point of $\edg_0$.

The following Proposition asserts that
the microlocal behavior \eqref{micronormal} of $w$ on $\edg_0$ uniquely determines that on all other edges, if we ``cut'' the periodic trajectory $\Gamma_1$ at the point $\rho_0$: 
Let $\ope{Path}(\rho_1,\rho_2)$ be the set of all generalized classical trajectories starting from $\rho_1$ and ending at $\rho_2$ without passing through $\rho_0$. This is an infinite set in general.

\begin{lemma}\label{lem:PAC}
Let $w$ be a resonant state satisfying \eqref{micronormal}.
Then at any point $\rho$ on $\Gamma\setminus \mathcal V$, different from $\rho_0$, we have the following microlocal expression of $w$:
\be\label{eq:PA-M}
w\equiv\alpha_\rho f_{\rho}\qtext{ near}\rho\qtext{with}
\alpha_{\rho}=\sum_{{\gamma}\in\ope{Path}(\rho_0,\rho)}\mathscr{P}({\gamma}).
\ee
The sum on the right hand side of \eqref{eq:PA-M}
absolutely converges  for $h$ small enough.
\end{lemma}

\begin{proof}
Put $\til{M}:=(I-E_{\edg_0,\edg_0})M$ where $E_{\edg_0,\edg_0}$ stands for the matrix unit. 
Then for $k\ge1$, each $(\edg,\edg')$-entry of $\til{M}^k$ stands for the sum of the probability amplitude associated with generalized classical trajectories from $\rho_{\edg'}$ to $\rho_{\edg}$ passing through $k$-crossing points without coming back to $\rho_0$:
\be\label{eq:tilMk}
(\til{M}^k)_{\edg,\edg'}=\sum_{\substack{\gamma\in\ope{Path}(\rho_{\edg'},\rho_\edg),\\\#\{t;\,\gamma(t)\in\mc{V}\}=k}}\mathscr{P}(\gamma).
\ee 
For $k$ large enough, such a generalized classical trajectory always intersects with both $\Gamma_1$ and $\Gamma_2$ since $\rho_0\in\Gamma_1$. This implies that the absolute value of each eigenvalue of $\til{M}$ is smaller than one, and $(I-\til{M})^{-1}=\sum_{k\ge0}\til{M}^k$. Then for any $\edg$, 
\be\label{eq:waf}
w\equiv \alpha_\edg f_{\rho_\edg}\qtext{ near}\rho_\edg
\ee
holds with $\alpha=(\alpha_\edg)_{\edg\in\mc{E}}$ given by
\ben
\alpha:=(I-\til{M})^{-1}\dl_{\edg_0},\quad\dl_{\edg_0}=(\dl_{\edg_0}(\edg))_{\edg\in\mc{E}}\qtext{(Kronecker's delta).}
\een
According to \eqref{eq:tilMk}, we have
\ben
\alpha_\edg=\sum_{\gamma\in\ope{Path}(\rho_0,\rho_\edg)}\mathscr{P}(\gamma),
\een
for any $\edg\in\mc{E}$. This with \eqref{eq:waf} shows the required formula \eqref{eq:PA-M} when $\rho$ is a point on an edge since the base point can be chosen arbitrarily. Apply \eqref{eq:multiPA} for otherwise.
\end{proof}

In the same way as in \cite{FMW3}, for a resonant state $w$ satisfying \eqref{micronormal}, we have
\be
\left\|w\right\|_{L^2(x_1,x_2)}^2=2\cA'(E)+\ord(h^{1/3}+h^{1/(m_0+1)}),
\ee
where the error $\ord(h^{1/3})$ comes from the turning point while $\ord(h^{1/(m_0+1)})$ comes from the crossing points (in \cite{FMW3}, the later is $\ord(h)$ $(m_0=1)$ and contained in the former). 
We finally obtain the formula
\be\label{eq:imzh}
\im z_h(E)=-\frac{h}{2\left|\cA'(E_0)\right|}\sum_{k=1,2}\Bigl|\sum_{\gamma\in\ope{Path}(\rho_0,\rho_k)}\mathscr{P}(\gamma)\Bigr|^2+\ord(h^{\frac{m_0+3}{m_0+1}+\kappa}),
\ee
where $\kappa=\min\{1/3,1/(m_0+1)\}$. 
This is of $\ord(h^{\frac{m_0+3}{m_0+1}})$ since each generalized classical trajectory changes between $\Gamma_1$ and $\Gamma_2$ at least one time. 

The principal terms in the sum of the right  hand side of \eqref{eq:imzh} come from the generalized classical trajectories
which change trajectory once from $\Gamma_1$ to $\Gamma_2$.
Therefore we have the following proposition.
\begin{proposition}\label{prop:widths}
The coefficient $\mc{D}$ of the principal term of the resonance width admits the formula
\be\label{eq:Conclusion}
\mc{D}(E,h)=\frac{h^{-2/(m_0+1)}}{2\left|\cA'(E_0)\right|}\sum_{k=1,2}\left|\sum_{\gamma\in\ope{Path}^1(\rho_0,\rho_k)}\mathscr{P}(\gamma)\right|^2,
\ee
where $\ope{Path}^1(\rho_0,\rho_k)$ is the set of $\gamma\in\ope{Path}(\rho_0,\rho_k)$  (see before Lemma \ref{lem:PAC} for the definition) which changes trajectory from $\Gamma_1$ to $\Gamma_2$ exactly once.
Here the sum over $k$ is for such a $k$ that the corresponding outgoing tail exists.
\end{proposition}
\begin{remark}
The constant $\mc{D}(E,h)$ defined by \eqref{eq:Conclusion}
is obviously independent of the choice of $\rho_0$, $\rho_1$ and $\rho_2$ because of the absolute value taken for the sum over the general classical trajectories with the same starting and ending points. This absolute value is of order $h^{1/(m_0+1)}$ because they change characteristic set once, and hence $\mc{D}(E,h)$ is bounded as $h\to 0^+$.
\end{remark}

\section*{Acknowledgement}
A part of this work was done during the stay of
the second and the third authors in VIASM (Vietnam Institute of Advanced Study in Mathematics). They thank VIASM for the financial support and the hospitality. The second author was also supported by the JSPS KAKENHI Grant Number 18K03384. The third author was supported by the Grant-in-Aid for JSPS Fellows Grant Number 22J00430. The first author acknowledges the financial support of the CONICYT FONDECYT Grant Number 3180390.


\begin{thebibliography}{50}
 
 % \bibitem[AgCo]{AgCo} 
 %J.~Aguilar, J.M.~Combes\,:
 %\newblock{\em A class of analytic perturbations for one-body Schr\"odinger Hamiltonians.}
 %\newblock{Comm. Math. Phys.,} 22 (1971), no. 4, 269--279.
 
 % \bibitem[Ba]{Ba} H.~Baklouti\,:
  % \newblock{\em Asymptotique des largeurs de r\'esonances pour un mod\`ele d'effet tunnel microlocal.}
   % \newblock{Ann. Inst. H. Poincare Phys. Theor.,} 68 (1998); no. 2, 179-228.
   
   % \bibitem[BM]{BM} P.~Briet, A.~Martinez\,:
  %\newblock{\em Estimates on the molecular dynamics for the predissociation process.}
   % \newblock{J. Spect. Theory,} 7 (2017); 487-517.
 
 
 
\bibitem{As} S.~Ashida\,:
\newblock{\em Molecular predissociation resonances below an energy level crossing.}
\newblock{Asymptot. Analysis}, vol. 107, no 3-4 (2018), pp. 135-167.
 
\bibitem{AsFu} M.~Assal, S.~Fujii\'e\,:
\newblock{\em Eigenvalue splitting of polynomial order for a system of Schr\"odinger operators with energy-level crossing.}
\newblock{Comm. Math. Physics } 386 (2021), pp. 1519-1550.

%\bibitem{AFH2} M.~Assal, S.~Fujii\'e, K. Higuchi\,:
%\newblock{\em Resonances generated by crossings of classical trajectories II: Type B crossings.} In preparation.


\bibitem{Ba} H.~Baklouti\,:
\newblock{\em Asymptotique des largeurs de r\'esonances pour un mod\`ele d'effet tunnel microlocal.}
\newblock{Ann. Inst. H. Poincare Phys. Th\'eor.}, 68 (1998); no. 2, 179-228.





%\bibitem[BFRZ1]
%{BFRZ1} J.-F.~Bony, S.~Fujii\'e, T.~Ramond, M.~Zerzeri\,:
%\newblock{\it Microlocal kernel of pseudodifferential operators at a hyperbolic fixed point.}
%\newblock{J. Funct. Anal.} 252 (2007) no. 1, 68--125.

%\bibitem[BFRZ2]
%{BFRZ2} J.-F.~Bony, S.~Fujii\'e, T.~Ramond, M.~Zerzeri\,:
%\newblock{\it Barrier-top resonances for non globally analytic potentials.}
%\newblock{J. Spectr. Theory,} 9 (2019), no. 1, 315--348.

\bibitem{BFRZbook} J.-F.~Bony, S.~Fujii\'e, T.~Ramond, M.~Zerzeri\,:
\newblock{\it Resonances for Homoclinic Trapped Sets.}
\newblock{Ast\'erisque Bofuraze}, Soci\'et\'e Math\'ematique de France, 2018.
 


%\bibitem{Cdv1} Y.~Colin de Verdi\`ere\,:
%\newblock{\em Bohr-Sommerfeld phases for avoided crossings.}
%\newblock{Preprint on arXiv 1103.1507}.

%\bibitem{Cdv2} Y.~Colin de Verdi\`ere\,:
%\newblock{\em The level crossing problem in semi-classical analysis I. The symmetric case.}
%\newblock{Ann. Institut Fourier} 53, no. 4 (2003), pp. 1023-1054.

\bibitem{Cdv3} Y.~Colin de Verdi\`ere\,:
\newblock{\em The level crossing problem in semi-classical analysis. II. The Hermitian case.}
\newblock{Ann. Institut Fourier} 359, no. 5 (2004), pp. 1423-1441.


 %\bibitem{CdvPa} Y.~Colin de Verdi\`ere, B.~Parisse\,:
%\newblock{\em \'Equilbre instable en r\'egime semi-classique: I-Concentration microlocale.}
%\newblock{Commun. In PDE,} 19(9-10), 1535-1563 (1994).


 \bibitem{DiSj} M.~Dimassi, S.~Sj\"ostrand\,:
\newblock{\em Spectral Asymptotics in the Semi-Classical Limit.}
\newblock{Cambridge University Press,} 1999.
 
  
\bibitem{DyZw} S.~Dyatlov, M.~Zworski\,:
\newblock{\it Mathematical Theory of Scattering Resonances. Graduate Studies in Mathematics, 200.}
\newblock{American Mathematical Soc.,} 2019.
 
  %\bibitem[FLM]{FLM} 
 % S.~Fujii\'e, A.~Lahamar-Benbernou, A.~Martinez\,:
  %\newblock{\em Width of shape resonances for non globally analytic potentials.}
  %\newblock{J. Math. Soc. Japan,} 63 (2011), no. 1, 1-78. 
  
   %\bibitem[FLN]{FLN} S.~Fujii\'e, C.~Lasser, L.~N\'ed\'elec\,:
   %\newblock{\em Semiclassical resonances for a two-level Schr\"odinger operator with a conical intersection.}
   %  \newblock{Asympt. Anal.,} 65 (2009), no. 1-2, 17-58.
  
    %\bibitem[Er]{Er} 
  %A.~Erd\'elyi\,:
   %\newblock{\em Asymptotic expansions.}
     %\newblock{Dover Publications} Inc., New York, 1956.
     
       % \bibitem{Fe} M.~V.~Fedoriuk\,:
  % \newblock{\em M\'ethode asymptotiques pour les Equations Diff\'erentielles Ordinaires.} 
    % \newblock{Ed. MIR} Moscow, 1987 (Asymptotic Analysis, Springer, 1993).

         \bibitem{FMW1} S.~Fujii\'e, A.~Martinez, T.~Watanabe\,:
   \newblock{\em Widths of resonances at an energy-level crossing I: Elliptic interaction.} \newblock{J. Diff. Eq.} 260 (2016), pp. 4051-4085.
     
     
            \bibitem{FMW2} S.~Fujii\'e, A.~Martinez, T.~Watanabe\,:
   \newblock{\em Widths of resonances at an energy-level crossing II: Vector field interaction.} 
     \newblock{J. Diff. Eq.} 262 (2017), pp. 5880-5895.
     
      \bibitem{FMW3} S.~Fujii\'e, A.~Martinez, T.~Watanabe\,:
   \newblock{\em Widths of resonances above an energy-level crossing.} 
\newblock{J. Funct. Anal.} 280, no. 6 (2021), 108918.

  
 % \bibitem{GM1} A.~Grigis, A.~Martinez\,:
   % \newblock{\em Resonance widths in a case of multidimensional phase space tunneling.}
     %\newblock{Asymptot. Anal.,} 91 (2015), no. 1, 33-904.
  
  %\bibitem{GM2} A.~Grigis, A.~Martinez, \,:
    %\newblock{\em Resonance widths for the molecular predissociation.} 
      %\newblock{ Anal. PDE,} 7 (2014), no. 5, 1027-1055.
 
 \begin{comment}
 
  \bibitem{FR} S.~Fujii\'e, T.~Ramond, \,:
    \newblock{\em Matrice de scattering et r\'esonanes associ\'ees \`a une orbite h\'et\'erocline.} 
      \newblock{ Anal. PDE,} 7 (2014), no. 5, 1027-1055.

 

 \bibitem{GG} C.~G\'erard, A.~Grigis, \,:
    \newblock{\em  Precise Estimates of Tunneling and Eigenvalues near a Potential Barrier.} 
      \newblock{J. Diff. Equ.,} 72 (1988), 149-177.
     
% \bibitem[Ha]{Ha} 
% E.M.~Harrell\,:
 %\newblock{\em General lower bounds for resonances in one dimension.}
 %\newblock{Comm. Math. Phys.,} 86 (1982), no. 2, 221-225.
 
% \bibitem[HeMa]{HeMa}
%B.~Helffer, A.~Martinez\,: 
%\newblock{\em Comparaison entre les diverses notions de r\'esonances.}
%\newblock{Helv. Phys. Acta,} 60 (1987), no. 8, 992-1003.
 
 \bibitem{HeSj1} 
 B.~Helffer, J.~Sj\"ostrand\,:
 \newblock{\em Multiple wells in the semiclassical limit I.}
 \newblock{Commun. Part. Diff. Eq,} (1984), 337-408.
 
 \bibitem{HeSj2} 
B.~Helffer, J.~Sj\"ostrand\,:
 \newblock{\em R\'esonances en limite semiclassique.}
 \newblock{Bull. Soc. Math. France, M\'emoire No. 24-25,} 1986.
 
 
\bibitem{Ho1}
L.~H\"ormander\,:  \newblock{\em  The analysis of linear partial differential operators I. Distribution theory and Fourier analysis.} 
\newblock{Springer-Verlag, Berlin 1983.}

\end{comment}



 \bibitem{HeSj3} 
 B.~Helffer, J.~Sj\"ostrand\,:
 \newblock{\em  Semiclassical analysis for Harper's equation. III. Cantor structure of the spectrum.} 
 \newblock{M\'em. Soc. Math. France, (1989), no.39 1-124.}


 \bibitem{Hi1} 
K.~Higuchi\,:
 \newblock{\em Resonances free domain for systems of Schr\"odinger operators above an energy-level crossing.} 
  \newblock{Rev. Math. Phys. \textbf{33} (2021), no.3 article no. 2150007}.
 
  \bibitem{Hi2} 
K.~Higuchi\,:
 \newblock{\em R\'esonances semiclassiques engendr\'ees par des croisements de trajectoires classiques}, 
\newblock{C. R . Acad. Sci. 359, issue 6 (2021), pp. 657-663.}.

% \bibitem{HW} 
%K.~Hirota, J.~Wittsten\,:
% \newblock{\em Complex eigenvalue splitting for the Dirac operator.}
 % \newblock{Commun. Math. Phys. 383, 1527-1558 (2021)}.


\bibitem{Ho1}
L.~H\"ormander\,:  \newblock{\em  The analysis of linear partial differential operators I. Distribution theory and Fourier analysis.} 
\newblock{Springer-Verlag, Berlin 1983.}

\bibitem{Hu} 
W. Hunziker\,:
 \newblock{\em Distortion analyticity and molecular resonance curves.} Ann. Inst. H. Poincar\'e Phys. Th\'eor.
45 (1986), no. 4, 339-358.





 \bibitem{ILR} 
A.~Ifa, H.~Louati, M.~Rouleux\,:
 \newblock{\em Bohr-Sommerfeld quantization rules revisited:
 The method of positive commutators.}
 \newblock{J. Math. Sci. Univ. Tokyo, 25} (2018), 91-127.

  

%\bibitem[HiSi]{HiSi} P.~Hislop, I.M.~Sigal\,:
%\newblock{\em Introduction to Spectral Theory. With Applications to Schr\"odinger Operators.}
%\newblock{Applied Mathematical Sciences, 113. Springer-Verlag, New York,} 1996.



\begin{comment}
 \bibitem{Jecko} 
T.~Jecko\,:
 \newblock{\em On the mathematical treatment of the Born-Oppenheimer approximation.}
J. Math. Phys. 55, 053504 (2014).



\bibitem{Klein} 
M.~Klein\,:
 \newblock{\em On the mathematical theory of predissociation.}
Ann. Phys. 178(1), 48-73 (1987).




\bibitem{KMSW} 
M.~Klein, A.~Martinez, R.~Seiler, X. W. Wang\,:
 \newblock{\em On the Born-Oppenheimer expansion for polyatomic molecules.}
Comm. Math. Phys. 143(3), 607-639 (1992).




\bibitem{Ku} 
V. V.~Kucherenko\,:
\newblock{\em  Asymptotics of the solution of the system $A(x,-ih\partial/\partial x)u=0$ as $h\to 0$.}
\newblock{Izv. Akad. Nauk SSSR Ser. Mat.,} vol. 8, no. 3 (1974), 631-666.

\bibitem{Kl} 
M.~Klein\,:
\newblock{\em On the mathematical theory of predissociation.} 
\newblock{Ann. Physics,} 178 (1987), no. 1, 48-73.

\bibitem{KMSW} 
M.~Klein, A.~Martinez, R.~Seiler, X.W.~Wang\,:
\newblock{\em  On the Born-Oppenheimer expansion for polyatomic molecules.}
\newblock{Comm. Math. Physics,} 143 (1992), no. 3, 607-639.


\bibitem{KK} 
A. V.~Krivko, V. V.~Kucherenko\,:
\newblock{\em  Semiclassical Asymptotics of the Matrix Sturm-Liouville Problem.}
\newblock{Mathematical Notes,} vol. 80, no. 1 (2006), 136-140.


\bibitem{LF} 
L. D.~Landau, E. M..~Lifshitz\,:
\newblock{\em  Quantum Mechanics: Non-Relativistic Theory.}
\newblock{} 


 
 
 
  \bibitem{Ma2} A.~Martinez\,:
  \newblock{\em A. Martinez: Estimates on complex interactions in phase space.}
Math. Nachr. 167 (1994), 203-254..

\end{comment}
 
  \bibitem{Ma} 
A.~Martinez\,:
 \newblock{\em An Introduction to Semiclassical and Microlocal Analysis.}
 \newblock{Springer-Verlag New-York, UTX Series,} 2002.
 
% \bibitem[MaMe]{MaMe} A.~Martinez, B.~Messirdi\,:
  % \newblock{\em Resonances for Diatomic Molecules in the Born-Oppenheimer Approximation.}
    % \newblock{Comm. Partial Differential Equations,} 19 (1994), no. 7-8, 1139-1162.
 
% \bibitem[MaSo]{MaSo} A.~Martinez, V.~Sordoni\,:
  % \newblock{\em Twisted peudodifferential calculus and application to the quantum evolution of molecules.}
    % \newblock{Mem. Amer. Math. Soc.,} 200 (2009), no. 936.
 
  \bibitem{Na} S.~Nakamura\,:  
   \newblock{\em On an example of phase-space tunneling.} 
   \newblock{Ann. Inst. H. Poincare Phys. Theor.,}  63 (1995), no. 2, 211-229.
   
\bibitem{Ol} Olver, F. W. J.\,:  
   \newblock{\em Asymptotics and Special Functions.}
   %Airy and related functions.} in
    \newblock{Academic Press, New York, 1974}.  %   \newblock{Olver, Frank W. J.; Lozier, Daniel M.; Boisvert, Ronald F.; Clark, Charles W., NIST Handbook of Mathematical Functions,}  Cambridge University Press, ISBN 978-0521192255, MR 2723248 
   
  %  \bibitem[Ra]{Ra} T.~Ramond\,:
 %\newblock{\em Intervalles d'instabilit\'e pour une \'equation de Hill \`a potentiel m\'eromorphe}
 %\newblock{bull. SOC. math. France} 121 (1993), 403-444
   
 % \bibitem{Pe} P.~Pettersson\,:
 %\newblock{\em WKB expansions for systems of Schr\"odinger operators with crossing eigenvalues.}
 %\newblock{Asymptotic Analysis}, 14 (1997) 1-48.
 
% \bibitem{Se} E.~Servat\,:
% \newblock{a}
 %\newblock{b}
 
  % \bibitem[Sj]{Sj} 
 %J.~Sj\"ostrand\,:
 %\newblock{\em }
 %\newblock{Ast\'erisque}
 
  \bibitem{Sj1} 
 J.~Sj\"ostrand\,:
 \newblock{\em Density of states oscillations for magnetic Schr\"odinger operators.}
 \newblock{Mathematics in Science and Engineering,} Volume 186, 1992, Pages 295-345.


   %\bibitem[Vo]{Vo} A.~Voros\,:
 %\newblock{\em The return of the quartic oscillator: the complex WKB method.}
 %\newblock{Ann. Inst. H. Poincare Sect. A,}  39 (1983), no. 3, 211-338.
 
 \bibitem{W} S.~Wakabayashi\,:
 \newblock{\em Singularities of solutions of the Cauchy problem for symmetric hyperbolic systems,}
 \newblock{Commun. Part. Diff. Eq.,} 9 (12): 1147--1177 (1984).
 

  \bibitem{Ya} D.R.~Yafaev\,:
 \newblock{\em The semiclassical limit of eigenfunctions of the Schr\"odinger equation and the Bohr-Sommerfeld quantization condition, revisited.}
 \newblock{Algebra i Analiz}, 22 (2010), no. 6, 270-291; 
 \newblock{translation in St. Petersburg Math. J.,} 22 (2011), no. 6, 1051-1067.
 
   \bibitem{Zw} M.~Zworski\,:
  \newblock{\em Semiclassical Analysis, Graduate Studies in Mathematics, 138.}
   \newblock{American Mathematical Soc.,} 2012.


\end{thebibliography}
\end{document}